\pdfoutput=1

\documentclass[conference]{IEEEtran}
\IEEEoverridecommandlockouts

\usepackage[utf8]{inputenc}

\usepackage{preamble}

\begin{document}

\title{Bitcoin-Enhanced Proof-of-Stake Security:\\ Possibilities and Impossibilities
\thanks{Contact author: DT}
}

\author{%
\IEEEauthorblockN{Ertem Nusret Tas}%
\IEEEauthorblockA{
  Stanford University\\
  nusret@stanford.edu%
}%
\and%
\IEEEauthorblockN{David Tse}%
\IEEEauthorblockA{%
  Stanford University\\
  dntse@stanford.edu%
}
\and%
\IEEEauthorblockN{Fangyu Gai}%
\IEEEauthorblockA{%
  BabylonChain\\
  fangyu.gai@babylonchain.io%
}%
\linebreakand%
\IEEEauthorblockN{Sreeram Kannan}%
\IEEEauthorblockA{%
  University of Washington, Seattle\\
  ksreeram@uw.edu%
}%
\and%
\IEEEauthorblockN{Mohammad Ali Maddah-Ali}%
\IEEEauthorblockA{%
  University of Minnesota\\
  maddah.ali.ee@gmail.com%
}%
\and%
\IEEEauthorblockN{Fisher Yu}%
\IEEEauthorblockA{%
  BabylonChain\\
  fisher.yu@babylonchain.io%
}%
}

\newtheorem{theorem}{Theorem}
\newtheorem{corollary}{Corollary}
\newtheorem{lemma}{Lemma}
\newtheorem{proposition}{Proposition}
\newtheorem{conjecture}{Conjecture}

\newtheorem{definition}{Definition}
\newtheorem{remark}{Remark}
\newtheorem{todo}{Todo}
\newtheorem{example}{Example}
\newtheorem{question}{Question}

\newcommand{\token}{BBL\xspace}
\newcommand{\chainnosp}{Bitcoin\ignorespaces}
\newcommand{\bpow}{Bitcoin\xspace}
\newcommand{\passive}{passive\xspace}
\newcommand{\bpos}{B-PoS\xspace}
\newcommand{\pos}{PoS\xspace}
\newcommand{\header}{\ensuremath{\mathsf{header}}}

\newcommand{\powchain}{\ensuremath{\mathcal{C}}}
\newcommand{\poschain}{\ensuremath{\mathcal{L}}}
\newcommand{\btc}{Bitcoin\xspace}
\newcommand{\nd}{\ensuremath{\mathsf{bd}}}
\newcommand{\turnover}{\ensuremath{\mathsf{tr}}}
\newcommand{\aux}{\ensuremath{\mathcal{CP}}}

\newcommand{\txroot}{\ensuremath{\mathsf{txr}}}
\newcommand{\ts}{\ensuremath{\mathsf{TS}}}
\newcommand{\gT}{\ensuremath{\mathsf{gT}}}
\newcommand{\sT}{\ensuremath{\mathsf{sT}}}
\newcommand{\LOG}{\ensuremath{\mathsf{Ledger}}}
\newcommand{\Ledger}[2]{%
    \ensuremath{\mathsf{PoSLOG}_{#1}^{#2}}
}
\newcommand{\Ledgerf}[2]{%
    \ensuremath{\mathsf{FinLOG}_{#1}^{#2}}
}
\newcommand{\POWLedger}[2]{%
    \ensuremath{\mathsf{PoWChain}_{#1}^{#2}}
}
\newcommand{\validator}{\ensuremath{\mathsf{v}}}
\newcommand{\client}{\ensuremath{\mathsf{c}}}
\newcommand{\Ttendermint}[0]{\ensuremath{T_{\mathrm{tm}}}}

\renewcommand{\algorithmiccomment}[1]{\hfill$\triangleright$ \emph{#1}}

\algnewcommand{\algorithmicswitch}{\textbf{switch}}
\algdef{SE}[SWITCH]{Switch}{EndSwitch}[1]{\algorithmicswitch\ #1\ \algorithmicdo}{\algorithmicend\ \algorithmicswitch}%
\algtext*{EndSwitch}%

\algnewcommand{\algorithmiccase}{\textbf{case}}
\algdef{SE}[CASE]{Case}{EndCase}[1]{\algorithmiccase\ #1}{\algorithmicend\ \algorithmiccase}%
\algtext*{EndCase}%

\algnewcommand{\algorithmicon}{\textbf{on}}
\algdef{SE}[ON]{On}{EndOn}[1]{\algorithmicon\ #1\ \algorithmicdo}{\algorithmicend\ \algorithmicon}%
\algtext*{EndOn}%

\algrenewcommand{\algorithmicdo}{}
\algrenewcommand{\algorithmicthen}{}

\algnewcommand{\algorithmicgoto}{\textbf{goto}}%
\algnewcommand{\Goto}[1]{\algorithmicgoto~\ref{#1}}%

\algnewcommand{\algorithmicbreak}{\textbf{break}}%
\algnewcommand{\Break}[0]{\algorithmicbreak}%

\algnewcommand{\algorithmiccontinue}{\textbf{continue}}%
\algnewcommand{\Continue}[0]{\algorithmiccontinue}%

\algnewcommand{\algorithmicwaiton}{\textbf{wait on}}%
\algnewcommand{\WaitOn}[1]{\algorithmicwaiton~{#1}}%

\newcommand*\circled[1]{\tikz[baseline=(char.base)]{
            \node[shape=circle,draw,inner sep=2pt] (char) {\footnotesize #1};}}

\newcommand{\AAD}{Availability-Accountability Dilemma\xspace}
\newcommand{\Aad}{Availability-Accountability Dilemma\xspace}
\newcommand{\aad}{availability-accountability dilemma\xspace}

\newcommand{\asr}{accountable safety resilience\xspace}
\newcommand{\Asr}{Accountable safety resilience\xspace}
\newcommand{\ASR}{Accountable Safety Resilience\xspace}

\newcommand{\alr}{accountable liveness resilience\xspace}
\newcommand{\Alr}{Accountable liveness resilience\xspace}
\newcommand{\ALR}{Accountable Liveness Resilience\xspace}

\newcommand{\ssr}{slashable safety resilience\xspace}
\newcommand{\Ssr}{Slashable safety resilience\xspace}
\newcommand{\SSR}{Slashable Safety Resilience\xspace}

\newcommand{\slr}{slashable liveness resilience\xspace}
\newcommand{\Slr}{Slashable liveness resilience\xspace}
\newcommand{\SLR}{Slashable Liveness Resilience\xspace}

\newcommand{\ourprotocol}{our protocol\xspace}
\newcommand{\Ourprotocol}{Our protocol\xspace}

\newcommand{\chlc}[2]{\ensuremath{\mathsf{ch}_{#1}^{#2}}}

\newcommand{\LOGda}[2]{%
    \ifthenelse{\equal{#1}{}}{%
        \ensuremath{\mathsf{LOG}_{\mathrm{da}}^{#2}}%
    }{%
        \ensuremath{\mathsf{LOG}_{\mathrm{da},#1}^{#2}}%
    }%
}
\newcommand{\LOGbft}[2]{%
    \ifthenelse{\equal{#1}{}}{%
        \ensuremath{\mathsf{LOG}_{\mathrm{bft}}^{#2}}%
    }{%
        \ensuremath{\mathsf{LOG}_{\mathrm{bft},#1}^{#2}}%
    }%
}
\newcommand{\LOGacc}[2]{%
    \ifthenelse{\equal{#1}{}}{%
        \ensuremath{\mathsf{LOG}_{\mathrm{acc}}^{#2}}%
    }{%
        \ensuremath{\mathsf{LOG}_{\mathrm{acc},#1}^{#2}}%
    }%
}
\newcommand{\tr}[2]{%
    \ifthenelse{\equal{#1}{}}{%
        \ensuremath{\mathsf{T}^{#2}}%
    }{%
        \ensuremath{\mathsf{T}_{#1}^{#2}}%
    }%
}
\newcommand{\wt}[2]{%
    \ifthenelse{\equal{#1}{}}{%
        \ensuremath{\mathsf{w}^{#2}}%
    }{%
        \ensuremath{\mathsf{w}_{#1}^{#2}}%
    }%
}

\newcommand{\LOGBLANKFIX}[0]{\ensuremath{\mathsf{LOG}}}

\newcommand{\PI}[0]{\ensuremath{\Pi}}
\newcommand{\PIlc}[0]{\ensuremath{\Pi_{\mathrm{lc}}}}
\newcommand{\PIbft}[0]{\ensuremath{\Pi_{\mathrm{bft}}}}
\newcommand{\PIacc}[0]{\ensuremath{\Pi_{\mathrm{acc}}}}

\newcommand{\Adv}[0]{\ensuremath{\mathcal A}}
\newcommand{\Env}[0]{\ensuremath{\mathcal Z}}

\newcommand{\ie}[0]{\emph{i.e.}\xspace}
\newcommand{\eg}[0]{\emph{e.g.}\xspace}
\newcommand{\cf}[0]{\emph{cf.}\xspace}
\newcommand{\Ie}[0]{\emph{i.e.}\xspace}
\newcommand{\Eg}[0]{\emph{E.g.}\xspace}
\newcommand{\wolog}[0]{w.l.o.g.\xspace}

\newcommand{\GST}[0]{\ensuremath{\mathsf{GST}}}
\newcommand{\GAT}[0]{\ensuremath{\mathsf{GAT}}}

\newcommand{\concat}[0]{\ensuremath{\mathbin\Vert}}
\newcommand{\T}[0]{\ensuremath{\mathcal{T}}}

\newcommand{\tx}[0]{\ensuremath{\mathsf{tx}}}
\newcommand{\txs}[0]{\ensuremath{\mathsf{txs}}}
\newcommand{\negl}[0]{\ensuremath{\operatorname{negl}}}
\newcommand{\Tconfirm}[0]{\ensuremath{T_{\mathrm{fin}}}}
\newcommand{\Tcheckpoint}[0]{\ensuremath{T_{\mathrm{cp}}}}
\newcommand{\Tbtcmode}[0]{\ensuremath{T_{\mathrm{btc}}}}
\newcommand{\Ttimeout}[0]{\ensuremath{T_{\mathrm{to}}}}
\newcommand{\Trecent}[0]{\ensuremath{T_{\mathrm{rec}}}}
\newcommand{\kcp}[0]{\ensuremath{k_{\mathrm{cp}}}}
\newcommand{\Tslot}[0]{\ensuremath{T_{\mathrm{slot}}}} %
\newcommand{\Cf}[0]{\ensuremath{\mathcal{C}}}
\newcommand{\Wf}[0]{\ensuremath{\mathcal{W}}}
\newcommand{\rd}[0]{\ensuremath{\mathsf{rd}}}
\newcommand{\Tstall}[0]{\ensuremath{T_\mathsf{st}}}

\newcommand{\bprop}[1]{%
    \ifthenelse{\equal{#1}{}}{%
        \ensuremath{\Hat{b}}%
    }{%
        \ensuremath{\Hat{b}_{#1}}%
    }%
}

\newcommand{\ld}[1]{%
    \ifthenelse{\equal{#1}{}}{%
        \ensuremath{\mathrm{L}^{(c)}}%
    }{%
        \ensuremath{\mathrm{L}^{(#1)}}%
    }%
}

\newcommand{\adj}[0]{\ensuremath{\mathcal{J}}}

\newcommand{\betaS}[0]{\ensuremath{\beta_{\mathrm{s}}}}
\newcommand{\betaA}[0]{\ensuremath{\beta_{\mathrm{a}}}}
\newcommand{\betaL}[0]{\ensuremath{\beta_{\mathrm{l}}}}

\newcommand{\fS}[0]{\ensuremath{f_{\mathrm{s}}}}
\newcommand{\fA}[0]{\ensuremath{f_{\mathrm{a}}}}
\newcommand{\fL}[0]{\ensuremath{f_{\mathrm{l}}}}

\newcommand{\CpReq}[3]{%
    \ifthenelse{\equal{#3}{}}{%
        \ensuremath{\langle\mathsf{#1},#2\rangle}%
    }{%
        \ensuremath{\langle\mathsf{#1},#2\rangle_{#3}}%
    }%
}

\newcommand{\clue}{message\xspace}
\newcommand{\clues}{messages\xspace}
\newcommand{\Clue}{Message\xspace}
\newcommand{\Clues}{Messages\xspace}

\newcommand{\AdvEnvSync}[0]{\ensuremath{(\Adv_{\mathrm{s}}, \Env_{\mathrm{s}})}}
\newcommand{\AdvEnvPsync}[0]{\ensuremath{(\Adv_{\mathrm{p}}, \Env_{\mathrm{p}})}}
\newcommand{\AdvEnvDA}[0]{\ensuremath{(\Adv_{\mathrm{da}}, \Env_{\mathrm{da}})}}
\newcommand{\AdvEnvPG}[0]{\ensuremath{(\Adv_{\mathrm{pda}}, \Env_{\mathrm{pda}})}}

\definecolor{myParula01Blue}{RGB}{0,114,189}
\definecolor{myParula02Orange}{RGB}{217,83,25}
\definecolor{myParula03Yellow}{RGB}{237,177,32}
\definecolor{myParula04Purple}{RGB}{126,47,142}
\definecolor{myParula05Green}{RGB}{119,172,48}
\definecolor{myParula06LightBlue}{RGB}{77,190,238}
\definecolor{myParula07Red}{RGB}{162,20,47}
\newcommand{\revI}[1]{\textcolor{blue}{#1}}
\newcommand{\revRe}[1]{\textcolor{purple}{#1}}

\makeatletter
\def\ps@headings{%
\def\@oddhead{\mbox{}\scriptsize\rightmark \hfil \thepage}%
\def\@evenhead{\scriptsize\thepage \hfil \leftmark\mbox{}}}
\newcommand{\linebreakand}{%
  \end{@IEEEauthorhalign}
  \hfill\mbox{}\par
  \mbox{}\hfill\begin{@IEEEauthorhalign}
}
\makeatother
\pagestyle{headings}

\maketitle

\begin{abstract}
     Bitcoin is the most secure blockchain in the world, supported by the immense hash power of its Proof-of-Work miners. Proof-of-Stake chains are energy-efficient, have fast finality but face several security issues: susceptibility to non-slashable long-range safety attacks, low liveness resilience and difficulty to bootstrap from low token valuation.  We show that these security issues are inherent in any PoS chain without an external trusted source, and propose a new protocol, Babylon, where an off-the-shelf PoS protocol checkpoints onto Bitcoin to resolve these issues. An impossibility result justifies the optimality of Babylon. A use case of Babylon is to reduce the stake withdrawal delay: our experimental results show that this delay can be reduced from weeks in existing PoS chains to less than 5 hours using Babylon, at a transaction cost of less than $\mathbf{10}$K USD per annum for posting the checkpoints onto Bitcoin.
\end{abstract}

\section{Introduction}
\label{sec:introduction}

\begin{figure*}
    \centering
    \includegraphics[width=\linewidth]{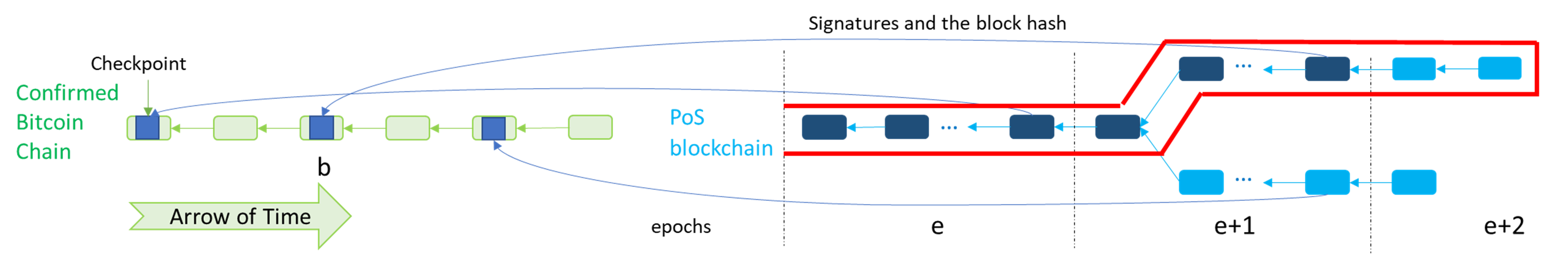}
    \caption{Babylon posts PoS block hashes and the validator signatures on them to Bitcoin. Ordering of these hashes enable clients to break ties between conflicting PoS chains, and slash adversarial validators before they withdraw after a safety violation.
    The canonical \pos chain in a client $\client$'s view is shown by the red circle.
    Dark blue blocks represent the checkpointed chain of \pos blocks in $\client$'s view.
    The fast finality rule determines the PoS chain, while the slow finality rule determines the checkpointed chain, which is always a prefix of the \pos chain.}
    \label{fig:checkpointing}
\end{figure*}

\begin{figure*}
    \centering
    \includegraphics[width=\linewidth]{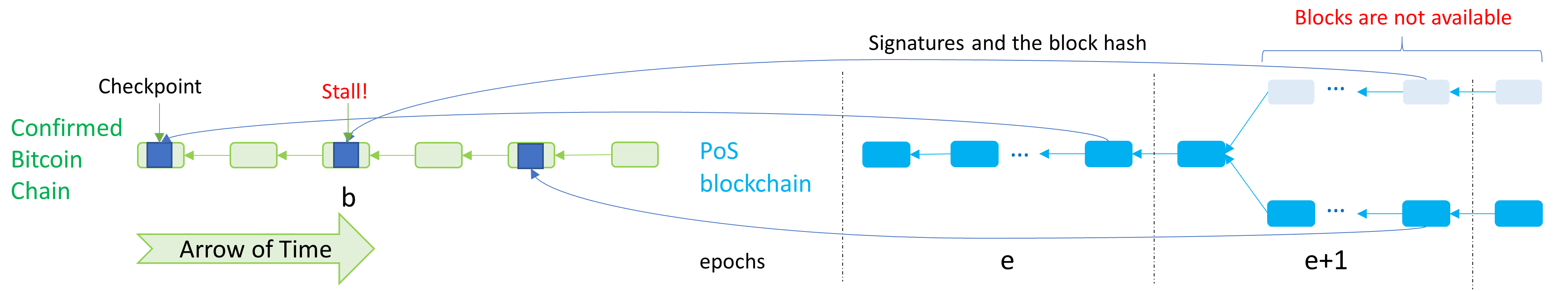}
    \caption{An adversary that controls a supermajority of active validators finalizes PoS blocks on an attack chain (top). It keeps the attack chain private, yet posts the hashes of the private blocks and their signatures on Bitcoin. 
    Once these checkpoints are deep in \btc, adversary helps build a conflicting chain (bottom) in public, and posts the hashes of its blocks and their signatures on Bitcoin.
    A client that sees the earlier checkpoint for the unavailable blocks, and the later one for the public blocks has two options:
    (1) It can stop outputting the new PoS blocks, or
    (2) it can ignore the earlier checkpoint and output the public blocks from the bottom chain. 
    However, the adversary can later publish the unavailable blocks, and convince a late-coming client to output the blocks from the top (attack) chain, causing a safety violation.
    Moreover, as the adversary might have withdrawn its stake by the time the blocks in the top chain are published, it cannot be slashed.
    To avoid this attack, clients choose to stall upon seeing block $b$, \ie \emph{emergency-break}, if they see a signed checkpoint for the unavailable blocks.}
    \label{fig:stalling}
\end{figure*}

\subsection{From Proof-of-work to proof-of-stake}

Bitcoin, the most valuable and arguably the most secure blockchain in the world, is supported by a proof-of-work (PoW) protocol that requires miners to compute many random hashes.
Many newer blockchain projects eschew the proof-of-work paradigm in favor of proof-of-stake (PoS).
A prominent example is Ethereum, which has migrated from PoW to PoS, a process that lasted $6$ years.
Other prominent PoS blockchains include single chain ecosystems such as Cardano, Algorand, Solana, Avalanche as well as multi-chain ecosystems such as Polkadot and Cosmos. The Cosmos ecosystem consists of many application-specific zones, all built on top of the Tendermint PoS consensus protocol~\cite{tendermint,tendermint_thesis}.

PoS protocols replace computational work with financial stake as the means to participate in the protocol.
To execute the protocol as \emph{validators}, nodes acquire coins of the PoS protocol, and lock up their stake as collateral in a contract.
This enables the \pos protocol to hold protocol violators accountable, and slash, \ie, burn their locked stake as punishment.

\subsection{Proof-of-stake security issues}
\label{sec:pos-insecurities}

Security of PoS protocols has traditionally been shown under the honest majority assumption, which states that the honest parties hold the majority of the stake \cite{snowwhite,badertscher2018ouroboros, algorand, tendermint}.
Introduced by Buterin and Griffith \cite{casper}, the concept of \emph{accountable safety} enhances the notion of security under honest majority with the ability to provably identify protocol violators
in the event of a safety violation.
Thus, accountable safety not only implies security under an honest majority, but also the identification of protocol violators if a large quorum of the validators are adversarial and cause a safety violation.
In lieu of making an unverifiable honest majority assumption, this approach aims to obtain a \emph{cryptoeconomic} notion of security by holding protocol violators accountable and \emph{slashing} their stake, thus enabling an exact quantification of the penalty for protocol violation.
This {\em trust-minimizing} notion of security is central to the design of PoS protocols such as Gasper \cite{gasper}, the protocol supporting PoS Ethereum, and Tendermint \cite{tendermint,tendermint_thesis}, supporting the Cosmos ecosystem.
However, there are several fundamental limitations to the security of the PoS protocols:

\noindent\textbf{Safety attacks are not slashable:}
    While a PoS protocol with accountable safety can identify attackers, slashing of their stake is not always possible, implying a lack of \emph{slashable safety}. For example, a posterior corruption attack can be mounted using old coins after the stake is withdrawn, and therefore cannot be slashed \cite{vitalik_weak_subj, snowwhite, badertscher2018ouroboros, long_range_survey}.
    These attacks are infeasible in a PoW protocol like Bitcoin as the attacker needs to counter the total difficulty of the existing longest chain.
    In contrast, they become affordable in a PoS protocol since the old coins have little value and can be bought by the adversary at a small price. 
    Posterior corruption is a long-known problem with PoS protocols, and several approaches have been proposed to deal with it under the honest majority assumption (Section \ref{sec:related-work}). 
    Theorem \ref{thm:pos-non-slashable} in Section \ref{sec:opt_safety} says that no PoS protocol can provide slashable safety without {\em external} trust assumptions. 
    A typical external trust assumption used in practice is {\em off-chain social consensus checkpointing}.
    As social consensus is a slow process, this type of checkpointing leads to a long \emph{withdrawal delay} after a validator requests to withdraw its stake (\eg, $21$ days for Cosmos zones \cite{cosmos_delay}), which reduces the liquidity of the system. Moreover, social consensus cannot be relied upon in smaller blockchains with an immature community.
    
\noindent\textbf{Low liveness resilience:}
    PoS protocols such as Snow White~\cite{sleepy,snowwhite} and Ouroboros~\cite{kiayias2017ouroboros,david2018ouroboros,badertscher2018ouroboros} guarantee liveness as long as the fraction of adversarial stake is below $1/2$.
    However, PoS protocols with accountable safety such as Tendermint and Gasper cannot ensure liveness beyond an adversarial fraction of $1/3$.
    This low liveness resilience of $1/3$ is a fundamental limitation of PoS protocols with accountable safety, even under synchrony~\cite[Appendix B]{forensics}.
    
\noindent\textbf{The bootstrapping problem:}
     Even if a PoS protocol provides slashable safety, the maximum penalty an adversary can suffer due to slashing does not exceed the value of its staked coins.
     Thus, the cryptoeconomic security of a PoS protocol is proportional to its token valuation.
    Many PoS chains, particularly ones that support one specific application (\eg, Cosmos zones) start with a low token valuation. This makes it difficult for new PoS chains to support high-valued applications (\eg, decentralized finance, NFTs). Similarly, a PoS chain that experiences a significant drop in token valuation becomes more susceptible to attacks.

\subsection{Leveraging external trust}
\label{sec:babylon-timestamping-service}

The main reason behind the security issues described above is the absence of a reliable {\em arrow of time}. 
For instance, posterior corruption attacks exploit the inability of the late-coming clients to distinguish between the canonical chain minted by the honest validators and the adversary's history-revision chain that is published much later \cite{snowwhite, long_range_survey}.
Hence, to guarantee a slashable notion of safety, PoS protocols need an external trust source that can periodically and publicly timestamp the canonical chain. Social consensus can be viewed as one such source of external trust, but because it is achieved off-chain, its level of security is hard to quantify. In this paper, we explore a more quantifiable approach, which is to use an {\em existing secure blockchain} as a source of external trust. Given such a trusted blockchain, we ask: {\em What is the limit to the security enhancement the trusted chain can provide to a PoS chain and what is the optimal protocol that achieves this limit?}

A natural example of such a trusted blockchain is Bitcoin. The main result of the paper is the construction of Babylon, where an off-the-shelf PoS protocol posts succinct information to Bitcoin for security enhancement. Moreover, we show that Babylon achieves the optimal security among all protocols that \emph{do not} post the entire PoS data to Bitcoin. Indeed, it is trivial to see that if the PoS protocol is allowed to post its entire data  onto the trusted chain, the PoS protocol can inherit its full security. But in a chain with low throughput like Bitcoin, posting the entire data is clearly infeasible. Our result shows exactly what the loss of security is from this limitation.

The idea of using a trusted parent chain to provide security to a PoS chain has been used in several industry projects and academic works. Most of these works focus on mitigating specific attack vectors. For example, a recently proposed protocol, BMS~\cite{bms}, uses Ethereum to keep track of the dynamic validator set of a PoS chain to withstand posterior corruption attacks. (That work was later extended to a protocol using Bitcoin in a concurrent work \cite{pikachu} to ours.) 
In our paper, we broaden the investigation to find out the best security guarantees a trusted public blockchain such as Bitcoin can provide to a PoS chain, and construct an optimal protocol, Babylon, that achieves these guarantees. A detailed comparison of Babylon and other approaches  is described in Table \ref{tab:just-btc} and Section \ref{sec:related-work}.

\subsection{Babylon}
\label{sec:protocol-specifics}

A PoS protocol is executed by \emph{validators}, which lock up their stake in a contract on the PoS chain. 
The design of Babylon specifies the kind of information validators post on Bitcoin and how this information is used by the clients, \ie, external observers of the protocol, to resolve attacks on the PoS chain (\cf Figure~\ref{fig:checkpointing}).
Its highlights are presented below:
\paragraph{Checkpointing}
\pos protocol proceeds in epochs during which the validator set is fixed.
The honest validators sign the hash of the last \pos block of each epoch (\cf Figure~\ref{fig:checkpointing}).
They subsequently post the hash and their signatures to \btc as \emph{checkpoints}.
Ordering imposed on these checkpoints by \btc enable the clients to resolve safety violations, and identify and slash adversarial validators engaged in posterior corruption attacks before they withdraw their stake.

\paragraph{Fast finality fork-choice rule (\cf Figure~\ref{fig:checkpointing})}
To output a \pos chain, a client $\client$ first identifies the confirmed prefix of the longest \btc chain in its view.
It then uses the sequence of checkpoints on \btc to obtain a \emph{checkpointed chain} of \pos blocks.
While constructing the checkpointed chain, \pos blocks with earlier checkpoints take precedence over conflicting blocks with later checkpoints.
Once $\client$ constructs the checkpointed chain, it obtains the full \pos chain by attaching the remaining \pos blocks that extend the checkpointed chain.
It stalls upon observing a fork among the later \pos blocks that extend the checkpointed chain.

Since \btc helps resolve earlier forks and obtain a unique checkpointed chain, safety can only be violated for recent \pos blocks in $\client$'s view.
Hence, adversarial validators cannot violate the safety of older \pos blocks through long range posterior corruption attacks after withdrawing their stake.
On the other hand, if a safety attack is observed for the recent \pos blocks, $\client$ can detect the adversarial validators and enforce the slashing of their stake.
Protocol thus ensures slashable safety.

\paragraph{Emergency break}
If the adversary controls a supermajority of the validators, it can sign the hashes of \pos blocks privately, and keep the blocks hidden from the clients. 
A client stops adding new blocks to its \pos chain if it observes a signed checkpoint on \btc, yet the corresponding block is unavailable.
This \emph{emergency break} is necessary to protect against data unavailability attacks (\cf Figure~\ref{fig:stalling}).
If the checkpoints consisted only of the block hashes, then the adversary could stall the \pos chain by sending a \emph{single} hash and pretending like it is the checkpoint of an unavailable \pos block.
Thus, the signatures on the checkpoints ensure that the adversary cannot cause an emergency break, unless it corrupts the supermajority of the current validator set.

\paragraph{Fallback to Bitcoin} 
If a transaction is observed to be censored, execution of the \pos protocol is halted, and the hashes of all future \pos blocks and the corresponding signatures on them are posted to \btc, which is directly used to order these blocks.
This is analogous to operating the \pos protocol as a \emph{rollup}, where \btc plays the role of the parent chain and the \pos validators act like sequencers.
A \pos chain that uses \btc directly to order its blocks is said to be in the \emph{rollup mode}.

\paragraph{Bitcoin safety \& slow finality fork-choice rule}
Clients can achieve Bitcoin safety for their \pos chains if they adopt a slow finality rule, where they only output the checkpointed chain in their views.
They wait until a \pos block or its descendants are checkpointed on \btc before outputting the block as part of the \pos chain.
In this case, the \pos chain is always safe (assuming Bitcoin is secure), however, its latency now becomes as large as \btc latency.

\setlength{\tabcolsep}{6pt}
\renewcommand{\arraystretch}{1.5}
\begin{table*}[]
\centering
\begin{tabular}{|c||c|c|c|c|c|}
  \hline
  & Safety & \multicolumn{3}{c|}{Liveness} & Withdrawal \\
  & & \multicolumn{3}{c|}{$f<n/3\quad\quad\quad\ $  $n/3 \leq f < n/2$  $\quad\quad\ f \geq n/2$} & \\
  \hline\hline
 KES~\cite{algorand} & $n/3$-safe & PoS Latency & No guarantee & No guarantee & ? \\
 \hline
 Pikachu~\cite{pikachu}/BMS~\cite{bms} & $n/3$-safe & PoS Latency & No guarantee & No guarantee & Bitcoin latency \\
 \hline
 Babylon: fast finality & $n/3$-slashably safe & PoS Latency & Bitcoin Latency & No guarantee & Bitcoin latency \\
 \hline
 Babylon: slow finality & Always safe  & Bitcoin Latency & Bitcoin Latency & No guarantee & Bitcoin latency \\
 \hline
\end{tabular}
\caption{The security guarantees of Babylon compared to other solutions, assuming the security of Bitcoin for Babylon and the security of Ethereum for BMS. Here, $f$ is the number of adversarial validators and $n$ is the total number of validators. $m$-safe means the protocol is safe whenever $f<m$, $m$-slashable-safe means that whenever safety is violated, $m$ validators can be slashed (which is a stronger property than $m$-safe). 
Stake withdrawals happen with Bitcoin latency on Babylon as long as liveness is satisfied, whereas it happens with Ethereum latency on BMS.
In theory, Algorand~\cite{algorand} can grant stake withdrawal requests in seconds as it uses key-evolving signatures (KES) to recycle keys after every signature, but since KES is highly incentive incompatible, Algorand still uses social consensus checkpointing.}
\label{tab:just-btc}
\end{table*}

\subsection{Security guarantees}
\label{sec:security-guarantees}

Table \ref{tab:just-btc} summarizes the security guarantees achieved by Babylon, assuming that Bitcoin is safe and live. Babylon resolves the three PoS security issues presented in Section \ref{sec:pos-insecurities} in the following way:

\noindent\textbf{Safety:} Under the fast finality rule, Babylon achieves slashable safety via checkpointing, and stalling whenever data is unavailable.
Slashable safety is not possible without an external source of trust. 
Moreover, Babylon achieves this with a stake withdrawal delay equal to the Bitcoin confirmation latency, \ie, the time for the timestamp of a withdrawal request to be confirmed on Bitcoin.
To estimate this delay, we implemented a checkpointing protocol and measured the confirmation latency for the checkpoints. See Section~\ref{sec:implementation}.

\noindent\textbf{Liveness:} 
Babylon improves the liveness resilience from $1/3$ to $1/2$ by using Bitcoin as a fallback. 
However, when the adversarial fraction exceeds $1/2$, liveness cannot be guaranteed. 
As shown by Theorem~\ref{thm:data-limit}, this is not Babylon's fault, but is inherent in any protocol that does not post the entire PoS transaction data to Bitcoin.
Then, the protocol is susceptible to data unavailability attacks.
In Section~\ref{sec:liveness}, we show one such attack on \emph{inactivity leak}, the method used by PoS Ethereum and Cosmos to slash inactive validators~\cite{eth_inactivity_leak,cosmos_inactivity_leak}.

\noindent\textbf{Bootstrapping:} Under the slow finality rule, Babylon is safe no matter how many adversarial validators there are, if Bitcoin is secure. Thus, Babylon achieves Bitcoin safety, albeit at the expense of Bitcoin confirmation latency.
Thus is useful in a bootstrapping mode or for important transactions, where slashable safety is not sufficient.

\subsection{Outline}

The rest of the paper is organized as follows. 
Section~\ref{sec:related-work} reviews related work.
Section~\ref{sec:model} presents the model and definitions of various security notions. 
In Section~\ref{sec:opt_safety}, we show that slashable safety is not possible for PoS chains without external trust.
We then present Babylon 1.0, a Bitcoin-checkpointing protocol that provides slashable safety.
In Section~\ref{sec:liveness}, we show the impossibility of liveness beyond a $1/2$ adversarial fraction of validators, even when there is a data-limited source of external trust.
We then improve Babylon 1.0 to the full Babylon protocol to provide the optimal liveness resilience.
In Section~\ref{sec:implementation}, we provide measurements for the confirmation latency of Bitcoin transactions containing checkpoints and demonstrate the feasibility of the Babylon protocol.
In Section~\ref{sec:btc-sec}, we describe the slow finality rule that provides Bitcoin safety to the PoS chains.

\section{Related Works}
\label{sec:related-work}

\subsection{Posterior corruption attacks} 

Among the PoS security issues discussed in Section \ref{sec:pos-insecurities}, posterior corruption attacks is the most well-known, \cite{vitalik_weak_subj, snowwhite, badertscher2018ouroboros, long_range_survey}.
In a posterior corruption attack also known as founders' attack, long range attack, history revision attack or costless simulation, adversary acquires the old keys of the validators after they withdraw their stake.
It then re-writes the protocol history by building a conflicting \emph{attack} chain in private.
The attack chain forks from the canonical one at a past block, where the old keys constituted a majority of the validator set.
Subsequently, the adversary replaces the old validators with new ones under its control, and reveals the attack chain to the clients observing the system at a later time.
The adversary thus causes clients to adopt conflicting chains at different times.

Several solutions have been proposed to mitigate the posterior corruption attacks:
1) checkpointing via social consensus (\eg, \cite{vitalik_weak_subj,snowwhite,winkle,barber2012}); 
2) use of key-evolving signatures (\eg, \cite{algorand,kiayias2017ouroboros,badertscher2018ouroboros}); 
3) use of verifiable delay functions (\eg, \cite{solana}); 
4) timestamping on an existing PoW chain like Bitcoin~\cite{pikachu}.

\subsubsection{Social consensus}
Social consensus refers to a trusted committee, distinct from the \pos validators, which periodically checkpoints the finalized PoS blocks on the canonical chain.
It attempts to prevent posterior corruption attacks by making the blocks on the attack chain distinguishable from the checkpointed ones on the canonical chain.
For instance, in PoS Ethereum, clients identify the canonical chain with the help of checkpoints received from their peers.
Since no honest peer provides a checkpoint on a private chain, posterior corruption attacks cannot confuse new validators~\cite{weak-subjectivity}.

As clients might receive their checkpoints from different peers, it is often difficult to quantify the trust assumption placed on social consensus.
A small committee of peers shared by all clients would imply centralization of trust, and make security prone to attacks targeting few nodes.
Conversely, a large committee would face the problem of reaching consensus on checkpoints in a timely manner, leading to long withdrawal delays.
For instance, Avalanche, Cosmos zones and PoS Ethereum have withdrawal delays of $14$, $21$ and $13$\footnote{This is calculated for $130,000$ attesters with an average balance of $32$ ETH to accurately model the targeted attester numbers on PoS Ethereum using~\cite[Table 1]{eth_delay}.} days respectively~\cite{avalanche_delay,cosmos_delay,eth_delay}.

\subsubsection{Key-evolving signatures (KES)}
KES requires the validators to forget old keys so that a posterior corruption attack cannot be mounted.
Security has been shown for various PoS protocols using key-evolving signatures under the honest majority assumption for the current, active validators~\cite{algorand,badertscher2018ouroboros}.
This assumption is necessary to ensure that the majority of the active validators willingly forget their old keys so that they cannot be given to an adversary at a later time.
However, there might be a strong incentive for the validators to record their old keys in case they later become useful.
Thus, KES render the honest majority assumption itself questionable by asking honest validators for a favor which they may be tempted to ignore for future gain.
This observation is formalized in Section~\ref{sec:long-range-attack}, which shows that KES cannot provide slashable safety for \pos protocols.

\subsubsection{Verifiable Delay Functions (VDFs)}
VDFs can help the clients distinguish the canonical chain generated a long time ago from an attack chain created much later, thus providing an arrow of time for the clients and protecting the \pos protocol from posterior corruption attacks.
However, like KES, VDFs standalone cannot provide slashable safety for \pos protocols (\cf Section~\ref{sec:long-range-attack}).
Another problem with VDFs is the possibility of finding faster functions~\cite{fastervdf}, which can then be used to mount a posterior corruption attack.

\subsubsection{Timestamping the validator set}

Posterior corruption attacks can be thwarted by timestamping the PoS validator set on an external public blockchain such as Ethereum~\cite{bms} and Bitcoin~\cite{pikachu}.
For instance, Blockchain/BFT Membership Service (BMS)~\cite{bms} uses an Ethereum smart contract as a \emph{reconfiguration service} that records the changes in the \emph{current} validator set.
When validators request to join or leave the current set, the existing validators send transactions containing the new validator set to the contract.
Upon receiving transactions with the same new validator set from sufficiently many existing validators (\eg, from over $1/3$ of the current validator set), the contract replaces the current set with the new one.

The goal of BMS is to protect the PoS protocol against posterior corruption attacks, where the adversary creates an attack chain using the credentials of old validators.
These old validators are then replaced by new, adversarial ones that are distinct from those on the canonical chain.
However, if the honest validators constitute over $2/3$ of the current validator set, the new adversarial validators on the attack chain cannot be validated and recorded by the contract before the new honest validators on the canonical chain.
Hence, BMS enables the late-coming clients to identify and reject the attack chain, providing security to PoS chains as long as the fraction of active adversarial validators is bounded by $1/3$ at any given time (\cf Table~\ref{tab:just-btc}).
By preventing posterior corruption attacks, BMS can also reduce the withdrawal delay of the PoS protocols from weeks to the order of minutes.

To prevent posterior corruption attacks, BMS requires an honest supermajority assumption on the current set of validators, and as such, does not provide slashable safety.
If the adversary controls a supermajority of the validator set, it can create an attack chain in private, simultaneous with the public canonical chain, and post the validator set changes of the private chain to the contract before the changes on the canonical chain.
Then, the late-coming clients could confuse the canonical chain for an attack chain, and identify its validators as protocol violators, which implies lack of slashable safety.

Unlike Babylon, BMS cannot ensure liveness if the fraction of adversarial active validators exceeds $1/3$ (\cf Table~\ref{tab:just-btc}).
Similarly, whereas Babylon can provide \btc safety for young \pos chains and important transactions, BMS cannot provide Ethereum safety to the constituent \pos protocols even by adopting a slow finalization rule.
This is because the BMS contract only keeps track of the \pos validators, and does not receive any information about the \pos blocks themselves such as block headers.

\subsection{Hybrid PoW-PoS protocols}

A PoS protocol timestamped by \btc is an example of a {\em hybrid PoW-PoS protocol}, where consensus is maintained by both the PoS validators and Bitcoin miners.
One of the first such protocols is Casper FFG, a finality gadget used in conjunction with a longest chain PoW protocol~\cite{casper}. 
The finality gadget is run by PoS validators as an overlay to checkpoint and finalize blocks in an underlay PoW chain, where blocks are proposed by the miners.
The finality gadget architecture is also used in many other PoS blockchains, such as PoS Ethereum~\cite{gasper} and Polkadot~\cite{stewart2020grandpa}.
Bitcoing timestamping can be viewed as a "reverse" finality gadget, where the miners run an overlay PoW chain to checkpoint the underlay PoS chains run by their validators.
In this context, our design that combines \btc with PoS protocols leverages off ideas from a recent line of work on secure compositions of protocols~\cite{ebbandflow,sankagiri_clc,acc_gadget}.

Combining a proof-of-work blockchain with classical BFT protocols was also explored in Hybrid consensus~\cite{hybrid} and Thunderella~\cite{thunderella}.
Hybrid consensus selects the miners of a PoW-based blockchain as members of a rotating committee, which then runs a permissioned and responsive BFT protocol that is secure up to $1/3$ adversarial fraction.
Unlike Babylon, Hybrid consensus cannot use Bitcoin as is, with no change to Bitcoin, and does not provide accountable or slashable safety.
Thunderella combines a responsive BFT protocol with a slow, longest chain protocol to achieve responsiveness and security up to $1/4$ and $1/2$ adversarial fraction respectively.
To improve the adversarial fraction tolerable for liveness from $1/3$ to $1/2$, Babylon uses insights from Thunderella~\cite{thunderella}.
However, unlike Babylon, Thunderella (in the PoW setting) cannot use Bitcoin as is and does not provide slashable safety.

\subsection{Timestamping}
Timestamping data on \btc has been used for purposes other than resolving the limitations of \pos protocols.
For instance, timestamping on \btc was proposed as a method to protect Proof-of-Work (PoW) based ledgers against $51\%$ attacks~\cite{karakostas-checkpointing}.
However, this requires the \btc network to contain \emph{observing} miners, which publish timestamps from the PoW ledger to be secured, only if the block data is available.
This implies changing \btc to incorporate data-availability checks for external ledgers, whereas in our work, we analyze the limitations of security that can be achieved by using Bitcoin as is.

Two projects that use \btc to secure \pos and PoW child chains are Veriblock~\cite{veriblock-whitepaper} and Komodo~\cite{komodo}. 
Both projects suggest checkpointing child chains on \btc to help resolve forks. 
However, they lack proper security proofs, and do not analyze how attacks on \pos chains can be made slashable.
Another use-case of timestamping is posting commitments of digital content to \btc to ensure integrity of the data~\cite{bitcoin-timestamp}.
In this vein, \cite{gipp15a} implements a web-based service to help content creators prove possession of a certain information in the past by posting its timestamps on \btc.

Other projects planning to use Bitcoin checkpointing include Filecoin, which aims to prevent posterior corruption attacks using the checkpoints, and BabylonChain that serves as a checkpoint aggregator for Cosmos zones to decrease the checkpointing cost in the presence of many chains~\cite{filecoin-checkpointing,babylonchain}.

\section{Model}
\label{sec:model}
\indent
\textbf{Notation.}
Given a positive integer $m$, $[m]$ denotes the set $\{1,2,\ldots,m\}$.
We denote \pos blocks by capital $B$ and the \btc blocks by lowercase $b$.
An event is said to happen with overwhelming probability if it happens except with negligible probability in the security parameter $\lambda$.

\textbf{Validators and clients.}
In the client-server setting of state machine replication (SMR), there are two sets of nodes: validators and clients.
Validators receive transactions as input, and execute a SMR protocol.
Their goal is to ensure that the clients obtain the same sequence of transactions, thus, the same end state.
We assume that the transactions are batched into \emph{blocks}, and the sequence obtained by the clients is a blockchain, denoted by $\poschain$.
Thus, we hereafter refer to the SMR protocols as \emph{blockchain protocols}.

To output a chain, clients collect consensus messages from the validators and download blocks.
Upon collecting messages from a sufficiently large quorum of validators, each client outputs a chain.
Clients can choose to output a chain at arbitrary times, and might be offline in between.
The set of clients include honest validators, as well as external nodes that observe the protocol infrequently.

The blockchain protocol has \emph{external validity}: 
A transaction in a given chain is valid with respect to its prefix if it satisfies external validity conditions.
A block is valid if it only contains valid transaction.
Clients output only the valid blocks.

\textbf{Blocks and chains.}
Each block consists of a block header and transaction data.
Block headers contain (i) a pointer (\eg, hash) to the parent block, (ii) a vector commitment (\eg, Merkle root) to the transactions data, and (iii) protocol related messages.
The total order across the blocks in a chain together with the ordering of the transactions by each block gives a total order across all transactions within the chain.

For a block $B$, we say that $B \in \poschain$, if $B$ is in the chain $\poschain$.
Similarly, we say $\tx \in \poschain$, if the transaction $\tx$ is included in a block that is in $\poschain$.
A block $B$ is said to \emph{extend} $B'$, if $B'$ can be reached from $B$ by following the parent pointers.
Conversely, the blocks $B$ and $B'$ are said to \emph{conflict} with each other if $B'$ does not extend $B$ and vice versa. 
The notation $\poschain_1 \prec \poschain_2$ implies $\poschain_1$ is a strict prefix of $\poschain_2$, whereas $\poschain_1 \preceq \poschain_2$ implies $\poschain_1$ is either a prefix of or the same as $\poschain_2$.
The chains $\poschain_1$ and $\poschain_2$ conflict with each other if they contain conflicting blocks.

\textbf{Environment and adversary.}
Transactions are input to the validators by the environment $\Env$.
Adversary $\Adv$ is a probabilistic polynomial time algorithm.
Before the protocol execution starts, adversary can corrupt a subset of validators, which are subsequently called \emph{adversarial}.
These validators surrender their internal states to the adversary and can deviate from the protocol arbitrarily (Byzantine faults) under the adversary's control.
The remaining validators are called \emph{honest} and follow the blockchain protocol as specified.
Time is slotted, and the validators are assumed to have synchronized clocks\footnote{Bounded clock offset can be captured as part of the network delay}.

\textbf{Networking.} 
Validators can broadcast messages to each other and the clients.
Messages are delivered by the adversary, which can observe a message before it is received.
Network is synchronous, \ie, the adversary delivers the messages sent by an honest validator to \emph{all} other honest validators and clients within $\Delta$ slots.
Here, $\Delta$ is a known parameter.
Upon joining the network at slot $t$, new validators and late-coming clients receive all messages broadcast by the honest validators before slot $t-\Delta$.

\textbf{Security.}
Let $\poschain^{\client}_r$ denote the chain obtained by a client $\client$ at slot $r$.
Let $\Tconfirm$ be a polynomial function of $\lambda$, the security parameter of the blockchain protocol.
We say that the protocol is $\Tconfirm$-secure if the following properties are satisfied:
\begin{itemize}
    \item \textbf{Safety:} For any slots $r,r'$ and clients $\client,\client'$, either $\poschain_{r}^{\client}$ is a prefix of $\poschain_{r'}^{\client'}$ or vice versa. 
    For any client $\client$, $\poschain^{\client}_{r}$ is a prefix of $\poschain^{\client}_{r'}$ for all slots $r$ and $r'$, $r' \geq r$.
    \item \textbf{$\mathbf{\Tconfirm}$-Liveness:} If $\mathcal{Z}$ inputs a transaction $\tx$ to an honest validator at some slot $r$, then, $\tx \in \poschain_{r'}^{\client}$ for any slot $r' \geq r+\Tconfirm$ and any client $\client$.
\end{itemize}

Let $f$ denote the upper bound on the number of adversarial validators over the protocol execution.
A \pos protocol provides $\fS$-safety if it satisfies safety whenever $f \leq \fS$.
Similarly, a \pos protocol provides $\fL$-$\Tconfirm$-liveness if it satisfies $\Tconfirm$-liveness whenever $f \leq \fL$.

\textbf{Accountable safety.}
To formalize accountable safety, we use the forensic analysis model by Sheng et al.~\cite{forensics}.
During normal execution, validators exchange messages (\eg, blocks or votes), and each validator records all protocol-specific messages received in an execution transcript.
If a client observes a safety violation, it invokes a forensic protocol by sending the conflicting chains to the validators.
Then, the honest validators send their transcripts to the client.
The forensic protocol takes these transcripts, and outputs a proof that identifies $f$ adversarial validators as protocol violators.
This proof is subsequently sent to all other clients, and serves as evidence that the identified validators have irrefutably violated the protocol rules.

\begin{definition}
\label{def:accountable-safety}
A blockchain protocol is said to provide accountable safety with resilience $f$ if when there is a safety violation, (i) at least $f$ adversarial validators are identified by the forensic protocol as protocol violators, and (ii) no honest validator is identified, with overwhelming probability.
Such a protocol provides \emph{$f$-accountable-safety}.
\end{definition}
$f$-accountable-safety implies $f$-safety, and as such is a stronger property.
Accountably-safe protocols considered in this paper have the same safety and \asr, though this is not necessarily true for all protocols.

\textbf{Proof-of-Stake (PoS) protocols.}
In a permissioned protocol, the set of validators stays the same over time.
In contrast, PoS protocols allow changes in the validator set.
In a \pos protocol, nodes lock up stake in a contract executed on the blockchain to become validators.
To distinguish the validators that are currently executing the protocol at a given time from the old validators, we will refer to the current validators as \emph{active}.
We assume that each active validator has the same stake, and is equipped with a unique cryptographic identity.
Once a validator becomes inactive, it immediately becomes adversarial if it has not been corrupted before, thus might engage in posterior corruption attacks.
Protocol execution starts with an initial committee of $n$ validators, and the contract allows at most $n$ validators to be active at any slot.
We assume that at all slots, there is a non-empty queue of nodes waiting to stake their coins, and each validator leaving the active set can be immediately replaced by a new validator\footnote{Many blockchains such as CosmosHub and Osmosis designate the top $n$ participants in the queue with the largest stake as the validator set~\cite{cosmoshub-validators,osmosis-numval}.}.

\pos protocols proceed in epochs starting at $1$ and measured in the number of PoS blocks.
For instance, if each epoch lasts $m$ blocks and a client observes a chain of $5m+3$ \pos blocks, then the first $m$ blocks belong to the first epoch, the second $m$ blocks to the second one, and so on, until the last $3$ blocks, which are part of the on-going epoch $6$.
During an epoch, the active validator set is fixed and the execution of the \pos protocol mimicks that of a permissioned blockchain protocol.
Across epochs, the \pos protocol supports changes in the active set through withdrawals.
An active validator can send a \emph{withdrawal} request to the protocol to leave the active set and retrieve its staked coin.
At the end of each epoch, clients inspect their chains and identify the validators whose withdrawal requests have been included in the chain.
Then, at the next epoch, these validators are replaced with new ones from the staking queue.

When a validator leaves the active set, its coin is \emph{not} necessarily released by the contract immediately.
Different \pos protocols can have different \emph{withdrawal delays}.
The withdrawal mechanism is central to security, and will be analyzed in later sections.
If a withdrawing validator's stake is first released in the view of a client $\client$ at slot $r$, the validator is said to have \emph{withdrawn its stake} in $\client$'s view at slot $r$.

\textbf{Model for Bitcoin.}
We model Bitcoin using the backbone formalism~\cite{backbone} and treat it as a black-box blockchain protocol, which accepts transactions and outputs a totally ordered sequence of \btc blocks containing these transactions.
To output the \btc chain \emph{confirmed} with parameter $k$ at slot $r$, a client $\client$ takes the longest chain of \btc blocks in its view, removes the last $k$ blocks, and adopts the $k$ deep prefix as its \btc chain at slot $r$.
We denote the \btc chain outputted by $\client$ at slot $r$ by $\powchain^{\client}_r$.
If a Bitcoin block $b$ or transaction $\tx$ first appears in the confirmed \btc chain (hereafter called the \btc chain) of a client $\client$ at slot $r$, we say that $\tx$ or $b$ has become \emph{confirmed} in $\client$'s view at slot $r$.
We say that Bitcoin is secure with parameter $k$ if it satisfies safety and $\Tconfirm$-liveness given the $k$-deep confirmation rule. 
Here, $\Tconfirm$ satisfies the following proposition:
\begin{proposition}[Chain Growth]
\label{prop:chain-growth}
Suppose \btc is secure with parameter $k$ with overwhelming probability.
Then, for any client $\client$, if a transaction $\tx$ is sent to \btc at slot $r$ such that $|\powchain^{\client}_{r-3\Delta}|=\ell$, $\tx \in \powchain^{\client}_{r'}$ for any $r' \geq r + \Tconfirm$, and $|\powchain^{\client}_{r + \Tconfirm}| \leq \ell+k$ with overwhelming probability.
\end{proposition}
If the adversarial fraction of the mining power is less than $1/2-\epsilon$ for some $\epsilon>0$, then there exists a parameter $k$ polynomial in $\lambda$ such that \btc is secure with parameter $k$ and satisfies the above proposition~\cite{backbone}.

\section{Optimal safety}
\label{sec:opt_safety}

We next formalize the concept of slashable safety, and analyze a simplified version of Babylon, called Babylon 1.0, which achieves the optimal safety guarantees.
The full protocol with the optimal liveness guarantees is presented in Section~\ref{sec:liveness}.

\subsection{Slashable safety}
\label{sec:slashable-safety}
A useful feature of the \pos protocols is the ability to impose financial punishments on the protocol violators through the slashing of their locked stake.
Slashable security extends the notion of accountability to \pos protocols.
A validator $\validator$ becomes \emph{slashable} in the view of a client $\client$ at slot $r$ if, 
\begin{enumerate}
    \item $\client$ has received or generated a proof through the forensic protocol by slot $r$ such that $\validator$ is irrefutably identified as a protocol violator,
    \item $\validator$ has not withdrawn its stake in $\client$'s view by slot $r$.
\end{enumerate}

In practice, once the contract that locks $\validator$'s stake receives a proof accusing $\validator$, it will attempt to slash $\validator$'s stake if it is not withdrawn yet.
However, if liveness is violated, the chain might not execute new transactions, and slash $\validator$'s stake.
Consequently, we opted to use the word `slashable' to indicate the conditional nature of \emph{slashing} on the resumption of chain activity after the security violation.
We discuss the conditions under which stake can be slashed inSection~\ref{sec:slashing}.

\begin{definition}
\label{def:slashable-safety}
A blockchain protocol is said to provide slashable safety with resilience $f$, if when there is a safety violation, at least $f$ adversarial validators become slashable in the view of all clients, and (ii) no honest validator becomes slashable in any client's view, with overwhelming probability.
Such a protocol provides $f$-slashable-safety.
\end{definition}

Slashable safety resilience of $f$ implies that in the event of a safety violation, all clients identify $f$ or more adversarial validators as protocol violators before they withdraw their stake, and no client identifies any honest validator as a protocol violator, with overwhelming probability.
By definition, \pos protocols that provide $f$-slashable-safety also provide $f$-accountable safety and $f$-safety.

\subsection{Slashable safety is not possible without external trust}
\label{sec:long-range-attack}
Without additional trust assumptions, no \pos protocol can provide slashable safety.
Suppose there is a posterior corruption attack, and a late-coming client observes two conflicting chains. 
As the client could not have witnessed the attack in progress, it cannot distinguish the attack chain from the canonical one.
Hence, it cannot irrefutably identify any active validator on either chain as a protocol violator.
Although the client might notice that the old validators that have initiated the attack violated the protocol rules by signing conflicting blocks, these validators are not slashable as they have already withdrawn their stake.
Hence, no validator becomes slashable in the client's view.
This observation is formalized by the following theorem, proven in Appendix~\ref{sec:appendix-slashable-safety}:
\begin{theorem}
\label{thm:pos-non-slashable}
Assuming common knowledge of the initial set of active validators, without additional trust assumptions, no \pos protocol provides both $\fS$-slashable-safety and $\fL$-$\Tconfirm$-liveness for any $\fS,\fL>0$ and $\Tconfirm<\infty$.
\end{theorem}

A related but different theorem by Daian et al. states that without additional trust assumptions, it is impossible to design a secure \pos protocol, even under the honest majority assumption for the active validators, due to posterior corruption attacks~\cite[Theorem 2]{snowwhite}.
However, key-evolving signatures (KES) have been shown to prevent these attacks and help guarantee security under honest majority~\cite{david2018ouroboros}.
In contrast, Theorem~\ref{thm:pos-non-slashable} states that external trust assumptions are needed to build \pos protocols with slashable safety, which is a stronger property than safety under honest majority.
To emphasize this point, the following paragraph presents an attack that illustrates how KES falls short of providing slashable safety despite mitigating posterior corruption attacks.

{\bf An attack on slashable safety given key-evolving signatures and VDFs.}
Suppose the adversary controls a supermajority of the active validators.
In the case of KES, the adversarial active validators record their old keys, and use them to initiate a posterior corruption attack after withdrawing their stake.
They can thus cause a safety violation, yet, cannot be slashed as the stake is withdrawn.
In the case of VDFs, the adversarial active validators again construct a private attack chain while they work on the canonical one, and run multiple VDF instances simultaneously for both chains.
After withdrawing their stake, they publish the attack chain with the correct VDF proofs, causing a safety violation without any slashing of their stake.

\subsection{Babylon 1.0 protocol with fast finality}
\label{sec:btc-protocol}
\begin{figure*}
    \centering
    \includegraphics[width=0.8\linewidth]{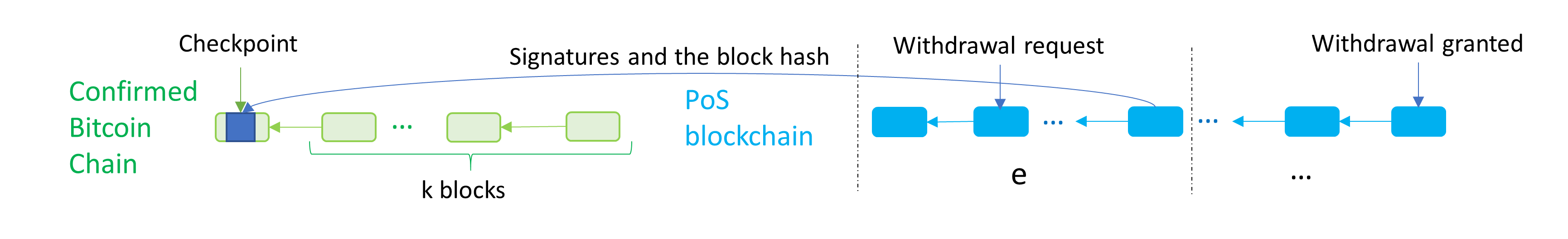}
    \caption{There is an epoch $e$ block containing a withdrawal request, and the hash of the last epoch $e$ block and the signatures from the corresponding active validator set appear in a confirmed \btc block in a client's view. The validator is granted permission to withdraw its stake once the \btc block with the checkpoint becomes $k$ deep in the confirmed \btc chains of sufficiently many validators.}
    \label{fig:withdrawal}
\end{figure*}

\setlength{\textfloatsep}{0.3cm}
\setlength{\floatsep}{0.3cm}
\begin{algorithm}
    \captionsetup{font=small} 
    \caption{The function used by a bootstrapping client $\client$ to find the canonical PoS chain $\mathcal{L}^\client_r$ at some slot $r$. It takes the blocktree $\T$ of finalized PoS blocks and the confirmed \btc chain $\powchain$ in $\client$'s view at slot $r$ as input, and outputs $\poschain^{\client}_r$. The function $\textsc{GetCkpts}$ outputs the sequence of checkpoints on the given Bitcoin chain $\powchain$ in the order of their appearance. The function $\textsc{GetActVals}$ takes a PoS chain $L$, an epoch number $\mathrm{ep}$, and outputs the active validators of epoch $\mathrm{ep}$ as determined by the given chain. The function $\textsc{IsValid}$ checks if the given checkpoint contains signatures by $2/3$ of the given active validators on its block hash, and if its epoch number matches the given number $\mathrm{ep}$. The function $\textsc{GetChain}$ returns the finalized chain within the given blocktree that ends at the preimage of the given block hash. It returns $\bot$ if a block at the preimage of the hash or its prefix chain is unavailable or not finalized. The function $\mathrm{ep}(.)$ returns the epoch of the last block within the input PoS chain. The function $\textsc{IsLast}$ returns true iff the given PoS chain ends at a block that is the last block of its epoch. The function $\textsc{GetChildren}$ returns the children of the given block within the blocktree. $L_0$ denotes the genesis PoS block and the letter $L$ is reserved for PoS chains.}
    \label{alg.bitcoin.checkpointing}
    \begin{algorithmic}[1]\small
    \Function{\sc OutputPosChain}{\T, \powchain}
        \Let{h_1,\ldots,h_m}{\textsc{GetCkpts}(\powchain)}
        \Let{\aux, \mathrm{ep}, \mathrm{actval}}{L_0, 1, \textsc{GetActVals}(L_0, 1)}
        \For{$i=1$ to $m$} \Comment{Obtain the checkpointed chain}
            \If{$\textsc{IsValid}(h_i,\mathrm{actval}, \mathrm{ep})$} \Comment{Check validity.}
                \label{line:isvalid}
                \Let{L_i}{\textsc{GetChain}(\T, h_i)}
                \If{$L_i = \bot$}
                    \label{line:btc2}
                    \State\Return $\aux$ \Comment{Stall: Data Unavailable}
                \ElsIf{$\aux \preceq L_i \land \mathrm{ep}(L_i)=\mathrm{ep}$}
                    \label{line:btc1}
                    \Let{\aux}{L_i}
                    \label{line:update}
                    \Comment{Update $\aux$.}
                    \If{$\textsc{IsLast}(L_i)$}
                        \Let{\rm ep}{ep+1}
                        \Let{\rm actval}{\textsc{GetActVals}(L_i, \mathrm{ep})}
                    \EndIf
                \EndIf
            \EndIf
        \EndFor
        \Let{\poschain, \mathrm{chs}}{\aux, \textsc{GetChildren}(\T,\aux[-1])}
        \label{line:btc3}
        \While{$|\mathrm{chs}|=1$}
            \Let{\poschain}{\poschain \concat \mathrm{chs}}~\Comment{Add the new child to $\mathcal{L}$}
            \Let{\rm chs}{\textsc{GetChildren}(\T,\mathrm{ch})}
        \EndWhile
        \State\Return $\poschain$ 
    \EndFunction
    \end{algorithmic}
\end{algorithm}
\setlength{\textfloatsep}{0.3cm}
\setlength{\floatsep}{0.3cm}

To provide \pos protocols with slashable safety, \btc can be used as the additional source of trust.
Babylon 1.0 is a checkpointing protocol which can be applied on any \pos blockchain protocol with accountable safety (\eg, PBFT~\cite{pbft}, Tendermint~\cite{tendermint}, HotStuff~\cite{yin2018hotstuff}, Streamlet~\cite{streamlet}, PoS Ethereum finalized by Casper FFG~\cite{casper,gasper}) to upgrade the accountability guarantee to slashable safety.
We describe Babylon 1.0 applied on Tendermint, which satisfies $n/3$-accountable safety.

Let $\client$ denote a client (\eg, a late-coming client, an honest validator), whose goal is to output the \emph{canonical} \pos chain $\poschain$ that is consistent with the chains of all other clients.
We assume that $\client$ is a full node that downloads the \pos block headers and the corresponding transaction data.
It also downloads the \btc blocks and outputs a \btc chain $\powchain$, confirmed with parameter $k$.
In the following description, \emph{finalized} blocks refer to the \pos blocks finalized by Tendermint, \ie, the blocks with a $2n/3$ quorum of pre-commit signatures by the corresponding active validator set.
A block finalized by Tendermint is not necessarily included in the canonical \pos chain $\poschain$ of a client.
For instance, when the adversary engages in a posterior corruption attack, conflicting \pos blocks might seem finalized, forcing the late-coming clients to choose a subset of these blocks as part of their canonical \pos chain.
Algorithm~\ref{alg.bitcoin.checkpointing} describes how clients interpret the \btc checkpoints to agree on the canonical chain.

\paragraph{Checkpointing the \pos chain}
At the end of each epoch, the honest active validators sign the hash of the last finalized \pos block of the epoch (the hash is calculated using a collision-resistant hash function).
Then, an honest active validator $\validator$ sends a \btc transaction called the \emph{checkpoint transaction}.
The transaction contains the hash of the block, its epoch and a quorum of signatures on the hash from over $2n/3$ active validators of the epoch.
These signatures can be a subset of the (pre-commit) signatures that have finalized the block within the \pos protocol.
We hereafter refer to the block hash, epoch number and the accompanying signatures within a checkpoint transaction collectively as a \emph{checkpoint}.

Suppose a client $\client$ observes multiple finalized and conflicting \pos blocks. 
As Tendermint provides $n/3$-accountable safety, $\client$ can generate a proof that irrefutably identifies $n/3$ adversarial \pos validators as protocol violators.
In this case, $\client$ posts this proof, called the \emph{fraud proof}, to Bitcoin.

\paragraph{Fork-choice rule} 
(Alg.~\ref{alg.bitcoin.checkpointing}, Figure~\ref{fig:checkpointing})
To identify the canonical \pos chain $\poschain^{\client}_r$ at some slot $r$, $\client$ first downloads the finalized \pos blocks and constructs a blocktree denoted by $\mathcal{T}$. 
The blocktree contains the quorums of pre-commit signatures on its finalized blocks.
Let $h_i$, $i \in [m]$, denote the sequence of checkpoints on Bitcoin, listed from the genesis to the tip of $\powchain^{\client}_r$, the confirmed \btc chain in $\client$'s view at slot $r$.
Starting at the genesis \pos block $B_0=L_0$, $\client$ constructs a \emph{checkpointed} chain $\aux^{\client}_r$ of finalized \pos blocks by sequentially going through the checkpoints.
For $i=1,\ldots,m$, let $L_i$ denote the chain of PoS blocks ending at the block (denoted by $B_i$) at the preimage of the hash within $h_i$, if $B_i$ and its prefix chain are available in $\client$'s view at slot $r$.

If a client observes a checkpoint $h_i$ with epoch $e$ on its Bitcoin chain at slot $r$, yet does not see by slot $r+\Delta$ all the block data of $L_i$ and the respective validator signatures finalizing these blocks, then blocks on chain $L_i$ are deemed to be \emph{unavailable} and \emph{not finalized} in the client's view.

Suppose $\client$ has gone through the sequence of checkpoints until $h_j$ for some $j \in [m]$, and obtained the chain $L$ ending at some block $B$ from epoch $e$ as the latest checkpointed chain based on the blocktree $\mathcal{T}$ and $h_1 \ldots, h_j$.
Define $\tilde{e}=e+1$ if $B$ is the last block of its epoch; and $\tilde{e}=e$ otherwise.
The checkpoint $h_{j+1}$ is said to be \emph{valid} (with respect to its prefix) if $h_{j+1}$ contains $\tilde{e}$ as its epoch, and over $2n/3$ signatures on its block hash by the active validators of $\tilde{e}$ (Alg.~\ref{alg.bitcoin.checkpointing}, Line~\ref{line:isvalid}).
\begin{enumerate}
    \item (Alg.~\ref{alg.bitcoin.checkpointing}, Line~\ref{line:btc1})
    If (i) $h_{j+1}$ is valid, (ii) every block in $L_{j+1}$ is available and finalized in $\client$'s view by the active validators of their respective epochs, (iii) the block $B_{j+1}$ is from epoch $\tilde{e}$, and (iv) $L \preceq L_{j+1}$, then $\client$ sets $L_{j+1}$ as the checkpointed chain.
    \item \textbf{Emergency Break:} (Alg.~\ref{alg.bitcoin.checkpointing}, Line~\ref{line:btc2}, Figure~\ref{fig:stalling}) If (i) $h_{j+1}$ is valid, and (ii) a block in $L_{j+1}$ is either unavailable or not finalized by its respective active validators in $\client$'s view, then $\client$ stops going through the sequence $h_j$, $j \in [m]$, and outputs $L$ as its final checkpointed chain\footnote{The client $\client$ knows the active validator set for all epochs $e \leq \tilde{e}$. Every block in its checkpointed chain $L$ is available in $\client$'s view.
    If $B$ is the last block of epoch $e$, $\client$ can infer the active validator set of epoch $e+1$ from $L$.}.
    This premature stalling of the fork-choice rule is necessary to prevent the data availability attack described by Figure~\ref{fig:stalling}.
    \item If both of the cases above fail, $\client$ skips $h_{j+1}$ and moves to $h_{j+2}$ and its preimage block as the next candidate.
\end{enumerate}
Unless case (2) happens, $\client$ sifts through $h_j$, $j \in [m]$, and subsequently outputs the checkpointed chain $\aux^{\client}_r$.
If (2) happens, $\client$ outputs $L$ as both the checkpointed chain $\aux^{\client}_r$ and the canonical \pos chain $\poschain^{\client}_r$.
However, if $\client$ had previously outputted finalized blocks extending $\aux^{\client}_r$ in its canonical \pos chain, it does not roll back these blocks and freezes its old chain.
Moreover, if (2) happens, after waiting for a period of $2\Delta$, $\client$ sends checkpoint transactions to \btc for \emph{all} the blocks extending the last checkpointed block on its $\poschain^{\client}_r$ to enforce slashable safety (\cf proof of Theorem~\ref{thm:btc-only-slashable-safety}).
If the block at the tip of $\poschain^{\client}_r$ is already checkpointed, as in the case of a bootstrapping late-coming client, then there is no need to send any more checkpoint transactions.
Although an online client might have to send checkpoints for all the blocks on $\poschain^{\client}_r$ from the last epoch, this event happens only when the adversary controls over $2/3$ of the validators and can post an unavailable checkpoint.

Finally, suppose $\client$ outputs a checkpointed chain with some block $B$ at its tip.
Then, starting at $B$, $\client$ traverses a path to the leaves of the blocktree (Alg.~\ref{alg.bitcoin.checkpointing}, Line~\ref{line:btc3} onward).
If there is a single chain from $B$ to a leaf, $\client$ outputs the leaf and its prefix chain as the \pos chain $\poschain^{\client}_r$.
Otherwise, $\client$ identifies the highest \pos block $B'$ (potentially the same as $B$) in the subtree of $B$, which has two or more siblings, and outputs $B'$ and its prefix chain as $\poschain^{\client}_r$.
Since $\client$ attaches the latest finalized \pos blocks to the tip of its \pos chain $\poschain^{\client}_r$, as long as there are no forks among finalized \pos blocks, the latency for transactions to enter $\poschain^{\client}_r$ matches the latency of the \pos protocol, hence rendering \emph{fast finality} to the protocol.

\paragraph{Stake withdrawals}
(Figure~\ref{fig:withdrawal})
To withdraw its stake, a validator $\validator$ first sends a \pos transaction called the \emph{withdrawal request}.
It is granted permission to withdraw in a client $\client$'s view at slot $r$ if
\begin{enumerate}
    \item The withdrawal request appears in a \pos block $B$ in $\client$'s checkpointed chain $\aux^{\client}_r$.
    \item Checkpoint of $B$ or one of its descendants appears in a \btc block that is at least $k$ deep in $\powchain^{\client}_r$.
    \item The client has not observed any fraud proof in $\powchain^{\client}_r$ that identifies $\validator$ as a protocol violator.
    Similarly, the client has not observed checkpoints for conflicting \pos blocks in $\powchain^{\client}_r$, that implicates $\validator$ in a protocol violation.
\end{enumerate}

Once the above conditions are satisfied in the validator $\validator$'s view, it sends a \emph{withdrawal transaction} to the \pos chain.
An honest active validator includes this transaction in its \pos block proposal if the above conditions are satisfied in its view.
Upon observing a \pos proposal containing a withdrawal transaction, the honest active validators wait for $\Delta$ slots before they sign the block.
Afterwards, they sign the proposal only if the above conditions are also satisfied in their views.
By synchrony, if the conditions are satisfied in an honest proposer's view at the time of proposal, then they are satisfied in the view of all honest active validators at the time of signing.
Thus, the $\Delta$ delay ensures that $\validator$'s transaction is finalized by the \pos chain despite potential, short-lived split views among the honest active validators.
Once the transaction is finalized, the on-chain contract releases $\validator$'s staked coin.

\paragraph{Slashing}
Suppose a validator $\validator$ has provably violated the protocol rules.
Then, the contract on the \pos chain slashes $\validator$'s locked coins upon receiving a fraud proof incriminating $\validator$ if the \pos chain is live and $\validator$ has not withdrawn its stake.

If a client $\client$ observes a fraud proof incriminating $\validator$ in its confirmed \btc chain $\powchain^{\client}_r$ at slot $r$, $\client$ does not consider $\validator$'s signatures on future checkpoints as valid.
It also does not consider $\validator$'s signatures as valid when verifying the finality of the \pos blocks \emph{checkpointed} on \btc for the first time after the fraud proof.
For instance, suppose $\client$ observes a checkpoint for a \pos block $B_j$ in its \btc chain.
While verifying whether $B_j$ and the blocks in its prefix (that have not been checkpointed yet) are finalized, $\client$ considers signatures only by the active validators that have \emph{not} been accused by a fraud proof appearing in the prefix of the checkpoint on \btc.

\subsection{Security analysis}
\label{sec:btc-protocol-analysis}
\begin{proposition}
\label{prop:auxiliary-chain}
Suppose \btc is safe with parameter $k$ with overwhelming probability.
Then, the checkpointed chains held by the clients satisfy safety with overwhelming probability.
\end{proposition}
Proof is provided in Appendix~\ref{sec:appendix-security-proofs}, and uses the fact that the safety of \btc implies consensus on the sequence of checkpointed blocks.

\begin{theorem}[Slashable Safety]
\label{thm:btc-only-slashable-safety}
Suppose \btc is secure with parameter $k$ with overwhelming probability, and there is one honest active validator at all times.
Then, the Babylon 1.0 protocol with fast finality (Section~\ref{sec:btc-protocol}) satisfies $n/3$-slashable safety with overwhelming probability.
\end{theorem}

\begin{theorem}[Liveness]
\label{thm:btc-only-liveness}
Suppose \btc is secure with parameter $k$ with overwhelming probability, and the number of adversarial active validators is less than $n/3$ at all times.
Then, the Babylon 1.0 protocol with fast finality (Section~\ref{sec:btc-protocol}) satisfies $\Tconfirm$-liveness with overwhelming probability, where $\Tconfirm=\Theta(\lambda)$. 
\end{theorem}

Proof sketches for Theorems~\ref{thm:btc-only-slashable-safety} and~\ref{thm:btc-only-liveness} are given in Appendix~\ref{sec:proof-sketches}, and the full proofs are in Appendix~\ref{sec:appendix-security-proofs}.
We extend these results to a restricted partial synchrony model in Appendix~\ref{sec:appendix-partial-synchrony}.

\section{Optimal liveness}
\label{sec:liveness}

Babylon 1.0 provides slashable safety to Tendermint.
However, the protocol guarantees liveness only when over $2/3$ of the active validators is honest.
We next explore how Babylon 1.0 can be improved to achieve optimal liveness.

\subsection{No liveness beyond $1/2$ adversarial fraction}
\label{sec:data-limit}
Our first result is that no \pos protocol can achieve a liveness resilience of $\fL \geq n/2$ even with the help of Bitcoin, or for that matter any timestamping service, unless the entirety of the PoS blocks are uploaded to the service.

\emph{Timestamping service.}
Timestamping service is a consensus protocol that accepts messages from the validators, and provides a total order across these messages.
All messages sent by the validators at any slot $r$ are outputted in some order determined by the service, and can be observed by all clients at slot $r+1$.
If a client queries the service at slot $r$, it receives the sequence of messages outputted by the service until slot $r$.
The service can impose limitations on the total size of the messages that can be sent during the protocol execution.

\begin{theorem}
\label{thm:data-limit}
Consider a PoS or permissioned protocol with $n$ validators in a $\Delta$ synchronous network such that the protocol provides $\fS$-accountable-safety for some $\fS>0$, and has access to a timestamping service.
Suppose each validator is given an externally valid input of $m$ bits by the environment $\mathcal{Z}$, but the number of bits written to the timestamping service is less than $m\lfloor n/2 \rfloor-1$.
Then, the protocol cannot provide $\fL$-$\Tconfirm$-liveness for any $\fL \geq n/2$ and $\Tconfirm<\infty$.
\end{theorem}

\begin{figure}
    \centering
    \includegraphics[width=\linewidth]{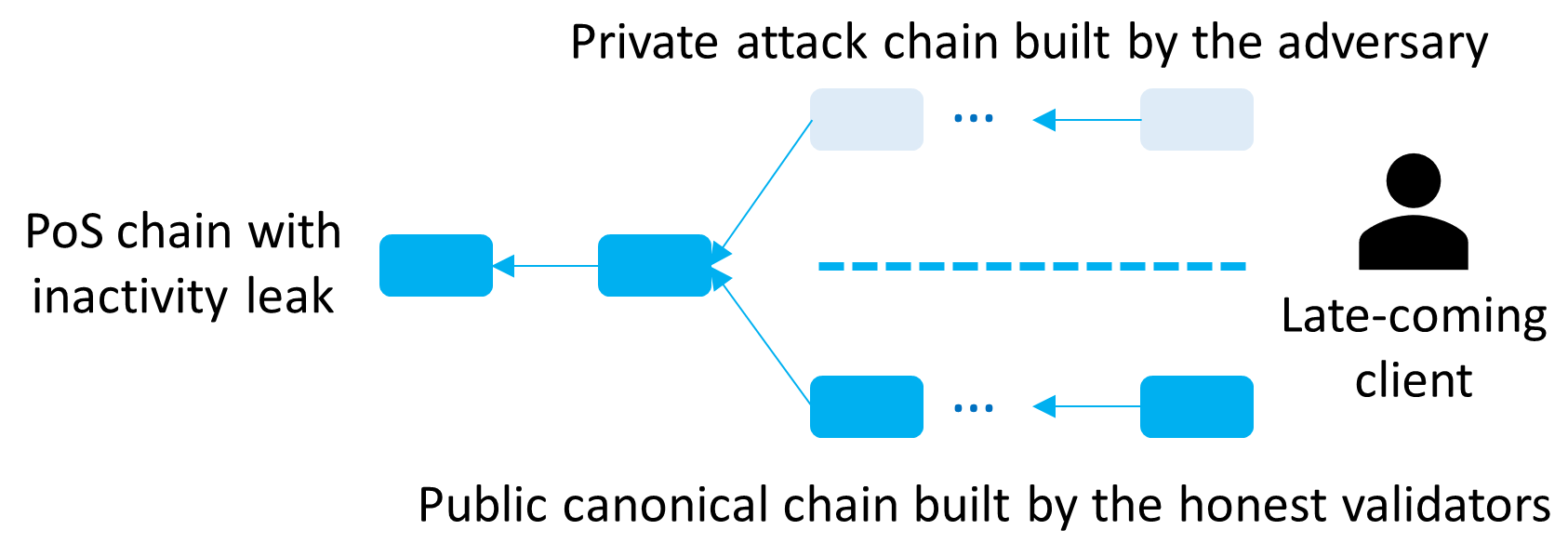}
    \caption{Inactivity leak attack. 
    Due to inactivity leak, honest and adversarial validators lose their stake on the attack and canonical chains respectively. A late-coming client cannot differentiate the canonical and attack chains.}
    \label{fig:inactivity-leak}
\end{figure}

To give intuition for the theorem, we analyze \emph{inactivity leak}, used by Cosmos zones~\cite{cosmos_inactivity_leak} and proposed for PoS Ethereum~\cite{eth_inactivity_leak} to slash inactive validators and recover from liveness attacks.
We show that when the adversary controls half of the validators, it can do a safety attack, where inactivity leak results in a gradual slashing of the {\em honest} validators' stake in the view of \emph{late-coming} clients.

Consider PoS Ethereum with the setup on Figure~\ref{fig:inactivity-leak}.
Half of the validators are adversarial, and build an attack chain that is initially kept private.
They do not communicate with the honest validators or vote for the blocks on the public, canonical chain.
As the honest validators are not privy to the adversary's actions, they also cannot vote for the blocks on the attack chain.
Since only half of the validators are voting for the public blocks, liveness of the finality gadget, Casper FFG, is temporarily violated for the canonical chain.
At this point, inactivity leak kicks in, and gradually slashes the stake of the adversarial validators on the canonical chain to recover liveness.
Similarly, the honest validators lose their stake on the private attack chain due to inactivity leak.
Finally, the adversary publishes the attack chain, which is subsequently adopted by a late-coming client.
As there are conflicting chains in different clients' views, this is a safety violation, and by accountable safety, the late-coming client must identify at least one validator as a protocol violator.
Since the client could not have observed the attack in progress, it cannot distinguish the attack chain from the canonical one.
As a result, with non-negligible probability, it identifies the honest validators on the canonical chain as protocol violators.
This contradicts with the accountability guarantee.

In the attack above, the data-limited timestamping service cannot help the late-coming client distinguish between the canonical and attack chains.
The honest validators cannot timestamp the entirety of public blocks on a data-limited service.
Thus, they cannot prove to a late-coming client that the canonical chain was indeed public before the attack chain was published.
This enables the adversary to plausibly claim to a late-coming client that the canonical chain was initially private, and its private attack chain was the public one.

The proof of Theorem~\ref{thm:data-limit} is given in Appendix~\ref{sec:appendix-liveness-resilience}, and generalizes the attack on inactivity leak to any \pos or permissioned protocol.
It exploits the indistinguishability of two worlds with different honest and adversarial validator sets, when the adversary controls over half of the validators and the timestamping service is data-limited.

\subsection{Optimal liveness: full Babylon protocol with fast finality}
\label{sec:bitcoin-with-honesty}
\begin{figure*}
    \centering
    \includegraphics[width=0.8\linewidth]{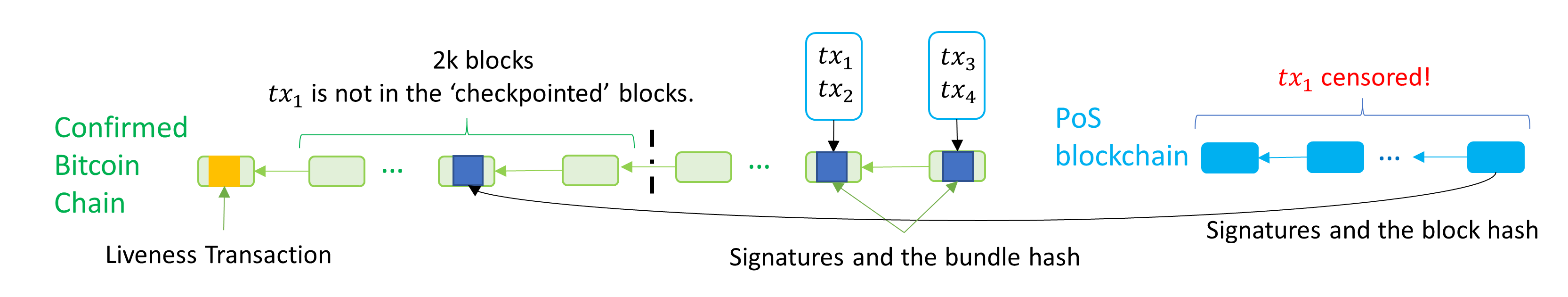}
    \caption{If $\tx_1$ is observed to be censored by an honest validator $\validator$, it sends a \emph{liveness transaction} to \btc. Once the liveness transaction becomes $2k$ deep in $\validator$'s and the clients' views, they enter the rollup mode. In the rollup mode, validators group transactions into bundles and post signed hashes of these bundles to \btc.}
    \label{fig:liveness-recovery}
\end{figure*}

\setlength{\textfloatsep}{0.3cm}
\setlength{\floatsep}{0.3cm} 
\begin{algorithm}
    \captionsetup{font=small} %
    \caption{The function used by a bootstrapping client $\client$ to find the canonical PoS chain $\poschain^{\client}_r$ at a given slot $r$. Here, $t_i$ denotes the type of the transaction on \btc, which can either be a checkpoint transaction, liveness transaction or a bundle. If $t_i$ is a checkpoint transaction or bundle, $h_i$ denotes the checkpoint, whereas if $t_i$ is a liveness transaction, $h_i$ denotes the hash of the censored transaction itself. The functions are defined in the caption of Algorithm~\ref{alg.bitcoin.checkpointing} except for $\textsc{GetHeight}$, which returns the height of the \btc block containing the given checkpoint, and $\textsc{isValB}$, which checks the validity of the bundle checkpoint. The Boolean variables $\mathrm{rmode}$ and $\mathrm{censor}$ become $\mathrm{True}$ if the checkpoints are associated with a rollup mode or the grace period respectively. The variable $\mathrm{censortx}$ denotes the censored transaction. The notation $a \in c$ is a Boolean, which returns $\mathrm{True}$ iff the chain $c$ contains a transaction whose hash is $a$.}
    \label{alg.liveness.recovery}
    \begin{algorithmic}[1]\small
    \Function{\sc OutputPosChain}{\T, \powchain}
        \Let{(t_1,h_1),\ldots,(t_m,h_m)}{\textsc{GetCkpts}(\powchain)}
        \Let{\aux, \mathrm{ep}, \mathrm{actval}}{L_0, 1, \textsc{GetActVals}(L_0, 1)}
        \Let{\rm rmode, censor, censortx, ht}{\mathrm{False},\mathrm{False},\bot,-1}
        \For{$i=1, \ldots, m$}
            \If{$\mathrm{censor} \land \textsc{GetHeight}(h_i) \geq \mathrm{ht} + 2k$}~\label{line:liveness1}
                \CommentLine{Enter rollup mode}
                \Let{\rm rmode, censor, censortx}{\mathrm{True}, \mathrm{False}, \bot}
            \ElsIf{$\mathrm{rmode} \land \textsc{GetHeight}(h_i) \geq \mathrm{ht} + 2k+\Tbtcmode$}~\label{line:liveness2}
                \CommentLine{Exit rollup mode}
                \Let{\rm rmode, ht}{\mathrm{False}, -1}
            \EndIf
            \CommentLine{Obtain the checkpointed chain}
            \If{$t_i = \mathrm{checkpoint} \land \neg\mathrm{rmode} \land \textsc{IsValid}(h_i,\mathrm{actval},\mathrm{ep})$}
                \CommentLine{Normal mode}
                \Let{L_i}{\textsc{GetChain}(\T, h_i)}
                \If{$L_i = \bot$}\label{line:unavailable1}
                    \State\Return $\aux$ \Comment{Stall: Data Unavailable}
                \ElsIf{$\aux \preceq L_i \land \mathrm{ep}(L_i) = \mathrm{ep}$}
                    \Let{\aux}{L_i}
                    \If{$\textsc{IsLast}(L_i)$}
                        \Let{\rm ep, actval}{ep+1, \textsc{GetActVals}(L_i, \mathrm{ep})}
                    \EndIf
                    \If{$\mathrm{censor} \land (\mathrm{censortx} \in \aux)$}~\label{line:not-censored}
                        \Let{\rm censor, censortx, ht}{\mathrm{False}, \bot, -1}
                    \EndIf
                \EndIf
            \ElsIf{$t_i = \mathrm{liveness} \land \neg\mathrm{rmode}$}      \Comment{Liveness transaction}
                \If{$(\tx \notin \aux) \land \neg\mathrm{censor}$}
                    \Let{\rm censor, censortx}{\mathrm{True},\{h_i\}}
                    \Let{\mathrm{ht}}{\textsc{GetHeight}(h_i)}
                \ElsIf{$(\tx \notin \aux) \land \mathrm{censor}$}
                    \Let{\mathrm{censortx}}{\mathrm{censortx} \cup \{h_i\}}
                \EndIf
            \ElsIf{$t_i = \mathrm{bundle} \land \mathrm{rmode} \land \textsc{IsValB}(h_i,\mathrm{actval},\mathrm{ep})$}\label{line:isvalidbundle}
                \CommentLine{Rollup mode}
                \Let{B_i}{\textsc{GetChain}(\T, h_i)}
                \If{$B_i = \bot$}\label{line:unavailable2}
                    \label{line:liveness4}
                    \State\Return $\aux$ \Comment{Stall: Data Unavailable}
                \Else
                   \label{line:liveness3}
                   \Let{\aux}{\aux \concat B_i} \Comment{Attach $B_i$ to $\aux$}
                \EndIf
            \EndIf
        \EndFor
        \Let{\poschain, \mathrm{chs}}{\aux, \textsc{GetChildren}(\T,\aux[-1])}
        \While{$|\mathrm{chs}|=1$}
            \If{$\mathrm{rmode} \lor (\mathrm{censor} \land |\powchain| \geq \mathrm{ht}+k)$}
                \label{line:liveness5}
                \Break
            \EndIf
            \Let{\poschain, \mathrm{chs}}{\poschain \concat \mathrm{chs}, \textsc{GetChildren}(\T,\mathrm{chs})}
        \EndWhile
        \State\Return $\poschain$
    \EndFunction
    \end{algorithmic}
\end{algorithm}
\setlength{\textfloatsep}{0.3cm}
\setlength{\floatsep}{0.3cm}

We now analyze whether \btc, despite its data limitation, can help increase the liveness resilience of Babylon 1.0 presented in Section~\ref{sec:btc-protocol}.
Babylon 1.0 provides $n/3$-slashable-safety and $n/3$-liveness; whereas Theorem~\ref{thm:data-limit} states that the liveness resilence of an accountably safe protocol cannot exceed $n/2$.
We next close this gap and achieve the optimal liveness resilience of $n/2$ for the full Babylon protocol.
Note that by \cite[Appendix B]{forensics}, liveness resilience of a \pos protocol that provides $n/3$-accountable-safety cannot exceed $n/3$ {\em in the absence of external trust}, so the improvement of Babylon's resilience from $n/3$ to $n/2$ depends crucially on the use of the data-limited timestamping service.
Indeed, if the adversary controls $f \in [n/3,n/2)$ of the active validators and violates liveness, the improved protocol uses \btc as a fallback mechanism to guarantee eventual liveness.

The full Babylon protocol proceeds in two modes: the normal mode and the rollup mode, where Bitcoin plays a more direct role in the ordering of the \pos blocks.
Execution starts and continues in the normal mode as long as no \pos transaction is censored.
If a transaction is observed to be censored, clients can force the execution to switch to the rollup mode. 
During the normal mode, checkpointing of the \pos chain, fork-choice rule, stake withdrawals and slashing work in the same way as described in Section~\ref{sec:btc-protocol}
(\cf Alg.~\ref{alg.liveness.recovery} for the complete fork-choice rule of the full Babylon protocol).
Thus, we next focus on the rollup mode.

\paragraph{Checkpointing}
If a transaction $\tx$, input to an honest validator by the environment $\mathcal{Z}$ at slot $r$, has not appeared in an honest validator $\validator$'s \pos chain within $\Ttendermint$ slots, \ie, is not in $\poschain^{\validator}_{r+\Ttendermint}$, $\validator$ sends a \emph{liveness transaction} to \btc.
Here, $\Ttendermint$ denotes the latency of Tendermint under synchrony.
The liveness transaction contains the hash of the censored $\tx$, and signals a liveness violation in $\validator$'s view.

Suppose there is a block $b$ within an honest validator $\validator$'s \btc chain, containing a liveness transaction for some censored $\tx$.
Upon observing $b$ become $k$ deep, $\validator$ sends a checkpoint transaction for its \pos chain, even if it has not reached the last block of its epoch.
If $\tx$ is not in its checkpointed chain, $\validator$ also stops executing Tendermint.
When $b$ becomes $2k$ deep in its \btc chain, if $\tx$ is still not in $\validator$'s checkpointed chain, $\validator$ enters the \emph{rollup mode}.
If $\tx$ appears in $\validator$'s checkpointed chain before $b$ becomes $2k$ deep in its \btc chain, then $\validator$ resumes its execution of Tendermint.
Honest validators agree on when to halt or resume the protocol execution as these decisions are based on the checkpoints on Bitcoin.

Once in the rollup mode, each honest validator collects transactions into \emph{bundles} that are broadcast to all other validators.
Upon observing a bundle of externally valid transactions, honest validators sign the hash of the bundle.
If an honest validator observes a bundle whose hash has been signed by over $n/2$ validators, it sends a checkpoint transaction containing the hash of the bundle, the last epoch number, and the signatures to \btc. 
We refer to the hash, epoch and the quorum of signatures collectively as a bundle checkpoint.

\paragraph{Fork-choice rule (Figure~\ref{fig:liveness-recovery}, Alg.~\ref{alg.liveness.recovery})}

Consider a client $\client$ that observes the protocol at some slot $r' \geq r+\Ttendermint$.
Suppose there is a block $b$ within $\powchain^{\client}_{r'}$, containing a liveness transaction for some censored $\tx$.
Once $b$ becomes $k$ deep in $\client$'s \btc chain, if $\tx$ is still not in $\client$'s checkpointed chain, $\client$ freezes its \pos chain.
At this point, $\client$ also sends a checkpoint transaction for its \pos chain, even if it has not reached the last block of its epoch.
Afterwards, $\client$ outputs new PoS blocks as part of its PoS chain, only if these new blocks are also part of its checkpointed chain (Alg.~\ref{alg.liveness.recovery}, Line~\ref{line:liveness5}). Note that if $\client$ was previously awake and has already outputted \pos blocks extending its checkpointed chain when the PoS chain was frozen, it does not roll back these blocks.

If $\client$ observes $\tx$ within its checkpointed chain by the time $b$ becomes $2k$ deep, it resumes outputting new blocks as part of its PoS chain, and continues the protocol execution in the normal mode(Alg.~\ref{alg.liveness.recovery}, Line~\ref{line:not-censored}).
Otherwise, once $b$ becomes $2k$ deep in $\client$'s \btc chain, if $\tx$ is not yet in $\client$'s checkpointed chain, it enters the \emph{rollup mode} (Figure~\ref{fig:liveness-recovery}, Alg.~\ref{alg.liveness.recovery}, Line~\ref{line:liveness1}).

Once in the rollup mode, $\client$ first constructs the checkpointed chain by observing the prefix of its \btc chain that ends at the $2k^{\text{th}}$ block extending $b$.
Suppose $B$ from epoch $e$ is the last \pos block appended to the checkpointed chain, and it is followed by a sequence $h_j$, $j \in [m]$,of bundle checkpoints in $\client$'s \btc chain.
Let $\tilde{e}=e+1$ if $B$ is the last block of epoch $e$, and $\tilde{e}=e$ otherwise.
The bundle checkpoint $h_{j}$ is said to be \emph{valid} (with respect to its prefix) if $h_{j}$ contains $\tilde{e}$ as its epoch, and over $n/2$ signatures on its block hash by the active validators of $\tilde{e}$ (Alg.~\ref{alg.liveness.recovery}, Line~\ref{line:isvalidbundle}).
After constructing the checkpointed chain, $\client$ sifts through $h_j$, $j \in [m]$, iteratively:
\begin{enumerate}
    \item (Alg.~\ref{alg.liveness.recovery}, Line~\ref{line:liveness3}) 
    If (i) $h_j$ is valid, and (ii) its preimage bundle is available in $\client$'s view, then $\client$ attaches the bundle to its checkpointed chain. During the rollup mode, clients output the same chain as the checkpointed and PoS chains.
    \item (Alg.~\ref{alg.liveness.recovery}, Line~\ref{line:liveness4}) If (i) $h_j$ is valid, and (ii) its preimage bundle is not available in $\client$'s view, then $\client$ stops going through the sequence $h_j$, $j \in [m]$, and returns its current checkpointed and PoS chains.
    \item If neither of the conditions above are satisfied, $\client$ skips $h_j$, and moves to $h_{j+1}$ as the next candidate.
\end{enumerate}

The client $\client$ leaves the rollup mode when it sees the $\Tbtcmode^{\text{th}}$ \btc block extending $b$ (\cf Alg.~\ref{alg.liveness.recovery} Line~\ref{line:liveness2}).
Here, $\Tbtcmode$ is a protocol parameter for the duration of the rollup mode.
After exiting the rollup mode, validators treat the hash of the last valid bundle checkpoint as the parent hash for the new \pos blocks, and execute the protocol in the normal mode.
It is possible to support stake withdrawals during the rollup mode, once the bundles containing withdrawal requests are timestamped by $k$ deep valid bundle checkpoints in \btc.

\subsection{Slashing and liveness after a safety violation}
\label{sec:slashing}
Theorem~\ref{thm:btc-only-slashable-safety} and its analogue below states that if \btc is secure, at least $n/3$ adversarial validators become slashable in the view of all clients when there is a safety violation.
However, the theorems do not specify whether these validators can be \emph{slashed} at all.
Indeed, slashing can only be done if the \pos chain is live, a condition that might not be true after a safety violation.

When blocks on conflicting chains are finalized in Tendermint, the chain with the earlier checkpoint in \btc is chosen as the canonical one.
However, the honest active validators might not be able to unlock from the conflicting chains they have previously signed and start extending the canonical one, as that would require them to sign conflicting blocks. 
Due to the absence of signatures from these stuck validators, the \pos chain might stall after a safety violation. 
Then, even though the validators that have caused the safety violation become slashable, the on-chain contract cannot slash them since the chain itself is not live.

This issue of liveness recovery after a safety violation is present in many BFT protocols, which strive to support accountability. 
For example, Cosmos chains enter into a panic state when safety is violated, and need a manual reboot based on social consensus with slashing done off chain~\cite{tendermint-panic}. 
With Bitcoin, however, this problem can be solved to an extent. 
As in the case of stalling or censorship, after which the protocol switches to the rollup mode, the honest validators can use \btc to unstuck from their respective forks, bootstrap the \pos chain, and slash the adversarial validators, if the honest validators constitute over \emph{half} of the active validator set.
They can use the same process described in Section~\ref{sec:bitcoin-with-honesty} for entering the rollup mode:
Once the \pos chain loses liveness, an honest validator posts a liveness transaction to \btc for the censored \pos transactions.
Soon afterwards, the honest \pos validators enter the rollup mode. 
This is because new checkpoints appearing on \btc and signed by the slashable validators will not be considered valid by the honest validators,
and cannot prevent the protocol from switching to the rollup mode.
Once in the rollup mode, with their majority, the honest \pos validators can sign new bundles, and send their checkpoints to \btc. 
Through these new bundles, they can finalize the censored transactions, and slash the adversarial validators that have previously become slashable.

After $n/3$ adversarial active validators are slashed, the remaining $n/2$ honest validators constitute a supermajority of the active validator set.
By treating the last valid bundle checkpoint as the new genesis block, they can switch back to the normal mode, and continue finalizing new \pos blocks through Tendermint (Alg.~\ref{alg.liveness.recovery}, Line~\ref{line:liveness2}).
Thus, the \pos chain can bootstrap liveness, and eventually return to the normal mode with fast finality after the safety violation.

\subsection{Analysis}
\label{sec:full-protocol-analysis}
\begin{theorem}[Slashable Safety]
\label{thm:btc-honest-majority-slashable-safety}
Suppose \btc is secure with parameter $k$ with overwhelming probability, and there is one honest active validator at all times.
Then, the Babylon protocol (Section~\ref{sec:bitcoin-with-honesty}) with fast finality satisfies $n/3$-slashable safety with overwhelming probability.
\end{theorem}
\begin{theorem}[Liveness]
\label{thm:btc-honest-majority-slashable-liveness}
Suppose \btc is secure with parameter $k$, and the number of adversarial active validators is less than $n/2$ at all times.
Then, the Babylon protocol (Section~\ref{sec:bitcoin-with-honesty}) with fast finality satisfies $\Tconfirm$-liveness with overwhelming probability, where $\Tconfirm$ is a polynomial in the security parameter $\lambda$.
\end{theorem}

Proof sketches for Theorems~\ref{thm:btc-honest-majority-slashable-safety} and~\ref{thm:btc-honest-majority-slashable-liveness} are given in Appendix~\ref{sec:proof-sketches}, and the full proofs are in Appendix~\ref{sec:appendix-security-proofs}.

\section{Experiments}
\label{sec:implementation}

Babylon allows a validator to withdraw its stake once the checkpoint of a \pos chain containing the withdrawal request becomes $k$ deep in the confirmed Bitcoin chain.
This implies a withdrawal delay on the order of Bitcoin's confirmation latency for the checkpoints, as opposed to the week long periods imposed by many \pos chains (\cf Section~\ref{sec:related-work}).
To measure the reduction in the withdrawal delay and to demonstrate the feasibility of using Bitcoin for checkpointing, we implemented a prototype \emph{checkpointer} for Tendermint, the consensus protocol of Cosmos zones.
We sent periodic checkpoints to the \btc Mainnet via the checkpointer to emulate the operation of Babylon 1.0 securing a Cosmos zone and measured the confirmation times of the checkpoints. 

\subsection{Checkpoints}
Our checkpointer utilizes Bitcoin's $\mathrm{OP\_RETURN}$ opcode, which allows $80$ bytes of arbitrary data to be recorded in an unspendable transaction~\cite{op-return} (\cf Bartoletti and Pompianu for an empirical study of the data within $\mathrm{OP\_RETURN}$ transactions~\cite{op-return-data}).
Checkpoints consist of an epoch number, the hash of a PoS block, and signatures by $2/3$ of the active validators of the epoch.
In Tendermint, each pre-commit signature is of $64$ bytes~\cite{tendermint-signatures}, whereas
prominent Cosmos zones such as Osmosis and Cosmos Hub feature $100$ to $175$ validators~\cite{osmosis-validators,cosmoshub-validators}.
However, if the checkpoints include signatures by $67$ validators, a single checkpoint would be approximately $4.3$ kBytes, requiring at least $54$ separate \btc transactions to record the information within \emph{one} checkpoint.

To reduce the number of transactions needed per checkpoint, Cosmos zones can adopt Schnorr~\cite{schnorr} or BLS~\cite{bls} signatures that allow signature aggregation.
In our experiments, we consider a Cosmos zone that uses BLS signatures.
To construct a checkpoint, we first generate $32$ bytes representing the block hash as specified by Tendermint's LastCommitHash~\cite{tendermint-signatures}, and $67$ BLS signatures on this value, each of $48$ bytes.
They emulate the pre-commit signatures by a $2/3$ quorum among $100$ validators.
We then aggregate these $67$ BLS signatures into a single one.
Thus, the checkpoints contain (i) an epoch number of $8$ bytes, (ii) a block hash of $32$ bytes, (iii) an aggregate BLS signature of $48$ bytes, and (iv) a bit map of $13$ bytes to help clients identify which validators' signatures were aggregated.
This implies a total size of $111$ bytes per checkpoint, reducing the number of transactions needed per checkpoint from $54$ to $2$.
In each of the two transactions, the $\mathrm{OP\_RETURN}$ also includes a fixed tag of $4$ bytes for the letters `BBNT' to signal the Babylon checkpoints.
This makes the total size of the two Bitcoin transactions $195$ and $183$ bytes respectively including other parts of the transactions.

\subsection{Measurements}
The withdrawal delay consists of two components: (i) the time to send the first checkpoint for a PoS chain including the withdrawal request, and (ii) the time for the checkpoint to become confirmed, \ie, $k$ deep in Bitcoin, after it is sent.
The first component equals the epoch length as a checkpoint is sent at the end of every epoch by an honest validator, and the time to prepare the checkpoint transactions is negligible compared to the epoch length.
This is because it takes $0.143$ ms to generate one BLS signature, and $0.042$ ms to aggregate $67$ BLS signatures on an Apple Macbook Pro M1 machine using the codebase~\cite{bls-library}.

To estimate the withdrawal delay of a Tendermint based PoS protocol using Babylon 1.0, we sent hourly checkpoints over the course of one day\footnote{Experiment was started at 1 am August 18, 2022, and run until 1 am August 19 (PT).}, for an epoch length of one hour.
At every hour, \emph{two} pairs of $\mathrm{OP\_RETURN}$ transactions for two identical checkpoints were broadcast to the Bitcoin network: one pair sent with a miner fee of $12$ Satoshis/byte, and the other sent with a fee of $3$ Satoshis/byte.
Given the transaction sizes, and Bitcoin's token price\footnote{$1\ \text{BTC}$ is taken as $23,563.83\ \text{USD}$, the maximum Bitcoin price observed on August 18, 2022~\cite{conversion-rate}.}, this implies a total cost of $1.07$ USD \emph{per checkpoint} at the fee level of $12$, and $0.27$ USD per checkpoint at the level of $3$ Satoshis/byte.
The fee levels were determined before the experiments with the help of a Bitcoin transaction fee estimator~\cite{btc-fee}, which displayed $12$ and $3$ Satoshis as the fee levels to have the transaction included within the next block, and the next $6$ blocks, respectively.

For each checkpoint associated with a different fee, we measured (i) the time $T_{k=6}$ for both transactions to be at least $k=6$ deep, and (ii) the time $T_{k=20}$ for both to be at least $k=20$ deep.
The confirmation depth $k=6$ was chosen following the conventional rule for Bitcoin.
The depth $k=20$ implies a low probability ($10^{-7}$) of safety violation for the blocks containing the checkpoints~\cite[Figure 2]{btc-latency}\footnote{This is for an adversary controlling $10\%$ of the hash rate, given a network delay of $10$ seconds for the Bitcoin network}.
We report the mean and standard deviation for these parameters across the $24$ checkpoints in Table~\ref{tab:measurements}\footnote{The data and the exact $\mathrm{OP\_RETURN}$ bytes can be accessed at \url{https://github.com/gitferry/Babylon-checkpoints}.}.

\begin{table}[h!]
\centering
\begin{tabular}{|c|c|c|c|c|}
\hline
Checkpointing cost per annum & $T_{k=6}$ (mins) & $T_{k=20}$ (mins) \\
\hline
$ 9373$ USD & $69.7 \pm 20.7$ & $192.7 \pm 26.1$ \\
\hline
$ 2365$ USD & $77.0 \pm 21.2$ & $204.1 \pm 31.2$ \\
\hline
\end{tabular}
\caption{Mean and standard deviation for $T_{k=6}$ and $T_{k=20}$ for the checkpoints sent over a day at two different fee levels.}
\label{tab:measurements}
\end{table}

Table~\ref{tab:measurements} shows that at an annual cost of less than $10,000$ USD, a PoS chain using Babylon can reduce its withdrawal delay from weeks to below $4$ hours.
The small cost of checkpointing is achieved via aggregate signatures, in the absence of which the cost would be on the order of millions.

\subsection{Cost of the full Babylon protocol}

To estimate the additional cost of checkpointing during the rollup mode, we consider a liveness transaction for some $\tx$.
The transaction uses the $\mathrm{OP\_RETURN}$ opcode to record the hash of $\tx$.
Upon observing the liveness transaction on Bitcoin, clients send checkpoints to \btc for their canonical \pos chains.
To avoid duplicate transactions for the same $\tx$ and duplicate checkpoints for the same \pos chain, a special validator can be tasked with posting liveness transactions, collecting \pos blocks from the clients and posting the checkpoint of the longest canonical \pos chain.
Acting as watch towers, clients observe the posted checkpoints and send checkpoints of higher or conflicting \pos blocks if the validator fails.
The number of checkpoints on Bitcoin before the rollup mode can thus be bounded while the special validator remains functioning.

Once in the rollup mode, honest validators periodically send bundle checkpoints to Bitcoin.
These checkpoints have the same structure as the checkpoints of \pos blocks.
The period of bundle checkpoints determine the cost of checkpointing during the rollup mode.
Frequent checkpoints reduce latency down to Bitcoin latency, whereas infrequent checkpoints result in a smaller cost.
The honest validators can designate a leader in round-robin fashion to create a bundle every hour, thus send one checkpoint per hour.
This makes the cost of checkpointing in the rollup mode equal to the normal mode.

\subsection{Checkpoint verification by light nodes}

In Sections~\ref{sec:btc-protocol} and~\ref{sec:bitcoin-with-honesty}, clients download all block data from Bitcoin to obtain the complete sequence of checkpoints.
To detect checkpoints, they check for the `BBNT' tag at the beginning of the $\mathrm{OP\_RETURN}$ data.
Although any client can send a Bitcoin transaction with the `BBNT' tag, transaction fees can deter spamming attacks, and miners can reject bursts of $\mathrm{OP\_RETURN}$ transactions.
An alternative solution is to leverage Bitcoin's Taproot upgrade to prevent a minority adversary from posting invalid checkpoints~\cite{pikachu}.

Clients also run a Bitcoin light node to avoid downloading invalid checkpoints.
Light nodes download block headers to find the confirmed Bitcoin chain and rely on full nodes to receive the checkpoints along with their (Merkle) inclusion proofs that are verified against the transaction roots in the headers~\cite{bitcoin}.
However, as the validity of checkpoints depends on the preceding ones and Bitcoin does not verify their validity, clients cannot validate a checkpoint based solely on its inclusion in a Bitcoin block.
Noticing that Bitcoin acts as a \emph{lazy blockchain} towards checkpoints, we can use insights from light nodes of lazy blockchains to help clients identify the canonical \pos chain without downloading all checkpoints and validating them sequentially~\cite{lazylight}.

\section{Babylon with slow finality: Bitcoin safety}
\label{sec:btc-sec}
So far in the paper, we have focused on the scenario where the clients of the PoS chain use the native \emph{fast finality rule}, where blocks are finalized upon gathering signatures from the PoS validators.
Since Bitcoin confirmation operates at a slower time scale, Bitcoin cannot protect the PoS chain against safety attacks under the fast finality rule. 
Bitcoin instead makes these attacks {\em slashable} by not allowing the attackers to withdraw stake after signing conflicting blocks.
However, clients might want Bitcoin level safety for important transactions or during the bootstrapping phase of a PoS chain, when slashability is not sufficient.
For this purpose, clients can choose to use a \emph{slow finality rule}, which requires a PoS block to be checkpointed by Bitcoin in addition to its finalization on the PoS chain. 
More specifically, a client $\client$ using the slow finality rule sets its \pos chain to be the same as its checkpointed chain at any time slot: $\poschain^{\client}_r = \aux^{\client}_r$.
A major drawback of this scheme is its lateny: client now waits until the checkpoints are $k$ deep in \btc, before it can output their blocks as its \pos chain.

\begin{corollary}
\label{cor:btc-sec-safety}
Suppose \btc is secure with parameter $k$ with overwhelming probability, and there is an honest active validator at all times.
Then, the Babylon protocol with slow finality satisfies safety with overwhelming probability.
\end{corollary}

Proof follows from Proposition~\ref{prop:auxiliary-chain}.
Corollary~\ref{cor:btc-sec-safety} holds for any number of adversarial active validators less than $n$.

\begin{corollary}
\label{cor:btc-sec-liveness}
Suppose \btc is secure with parameter $k$ with overwhelming probability, and the number of adversarial active validators is less than $n/2$ at all times.
Then, the Babylon protocol with slow finality satisfies $\Tconfirm$-liveness with overwhelming probability, where $\Tconfirm$ is a polynomial in the security parameter $\lambda$. 
\end{corollary}

Proof sketch for Corollary~\ref{cor:btc-sec-liveness} is given in Appendix~\ref{sec:proof-sketches}, and the full proof is in Appendix~\ref{sec:appendix-security-proofs}.

\section{Conclusion}
\label{sec:conclusion}
In this paper, we have constructed Babylon, where an accountable \pos chain uses Bitcoin as a timestamping service, and operates in two nodes: 
In the normal mode with \emph{fast finality}, clients of the \pos chain use the native finalization rule of the chain, and the role of Bitcoin is to provide slashable safety.
In the rollup mode, clients use a slow finality rule, and rely on Bitcoin as the consensus layer. Clients can choose which mode they want to operate at.
When there is censorship or stalling, the \pos chain can also switch to the rollup mode, and rely on Bitcoin to regain liveness. 
We have shown Babylon's optimality by characterizing the limitations of Bitcoin as a timestamping service. Just as Babylon accommodates a general set of accountable \pos protocols, the advantages and limitations provided for these protocols by Bitcoin are not specific to Bitcoin, but can be offered by any trusted public blockchain.
Hence, we can replace Bitcoin with any other trusted public blockchain that allows checkpointing \pos data, and all of the claimed improvements in the consensus security of \pos protocols would carry over with minor changes to the Babylon protocol.

\section*{Acknowledgements}
We thank Shresth Agrawal, Kamilla Nazirkhanova, Joachim Neu, Lei Yang and Dionysis Zindros for several insightful discussions on this project.
We thank Dionysis Zindros also for his help with the experiments.

Ertem Nusret Tas is supported by the Stanford Center for Blockchain Research, and did part of this work as part of an internship at BabylonChain Inc.
David Tse is a co-founder of BabylonChain Inc.

\bibliographystyle{plain}
\bibliography{references}

\appendices

\section{Comparison of Babylon with Bitcoin Improvement Proposals and Payment Channels}

Payment channels is a layer-2 scalability solution on Bitcoin that enables processing payments among multiple parties without using Bitcoin except in the case of adversarial participants~\cite{payment-channel-wattenhofer,payment-channel-lightening,payment-channel-eltoo}.
Unlike payment channels, Babylon uses Bitcoin to order checkpoints from PoS chains, while the transaction logic lives solely on the PoS chain.
Any existing PoS chain can opt in to use Babylon, with no change to Bitcoin.
Thus, Babylon preserves the responsiveness (fast finality) and high throughput of PoS chains without any upgrade to Bitcoin. 
In contrast, BIPs for offline transactions require changes to Bitcoin and cannot directly leverage off the properties of existing protocols.

\section{Proof Sketches for the Security Theorems}
\label{sec:proof-sketches}

\subsection{Proof Sketch for Theorem~\ref{thm:btc-only-slashable-safety}}

By the $n/3$-accountable safety of the \pos protocol, at least $n/3$ adversarial validators are identified after a safety violation.
These violations can only happen due to short-range forks at the tip of the clients' \pos chains, since the checkpointed chains do not conflict with each other by Proposition~\ref{prop:auxiliary-chain}. 
As the conflicting blocks within the short-range forks are not checkpointed by sufficiently deep Bitcoin blocks, the identified adversarial validators could not have withdrawn their stake in any clients' view.
Thus, at least $n/3$ validators become slashable in all clients' views after the safety violation.
The full proof of Theorem~\ref{thm:btc-only-slashable-safety} is presented in Appendix~\ref{sec:appendix-security-proofs}.

\subsection{Proof Sketch for Theorem~\ref{thm:btc-only-liveness}}

When the number of adversarial active validators is less than $n/3$, the \pos protocol satisfies safety and $\Tconfirm$-liveness.
Then, no unavailable checkpoint gathers enough signatures to be valid, preventing the adversary from triggering the stalling condition (Line~\ref{line:btc2}, Alg.~\ref{alg.bitcoin.checkpointing}).
The full proof of Theorem~\ref{thm:btc-only-liveness} is presented in Appendix~\ref{sec:appendix-security-proofs}.

\subsection{Proof Sketch for Theorem~\ref{thm:btc-honest-majority-slashable-safety}}

Since Bitcoin is responsible for ordering the bundles in the rollup mode and is assumed to be secure, there cannot be safety violations due to conflicting bundles. Thus, safety violations can only happen due to short-range forks with \pos blocks, in which case slashability of the adversarial validators follows from the proof of Theorem~\ref{thm:btc-only-slashable-safety}.
The full proof of Theorem~\ref{thm:btc-honest-majority-slashable-safety} is presented in Appendix~\ref{sec:appendix-security-proofs}.

\subsection{Proof Sketch for Theorem~\ref{thm:btc-honest-majority-slashable-liveness}}

When the number of adversarial active validators is less than $n/2$, the protocol satisfies liveness during the rollup mode, since the bundle checkpoints require signatures from half of the validators.
Moreover, no unavailable checkpoint gathers enough signatures to be valid, preventing the adversary from triggering the stalling condition (Lines~\ref{line:unavailable1} and~\ref{line:unavailable2}, Alg.~\ref{alg.bitcoin.checkpointing}).
The full proof of Theorem~\ref{thm:btc-honest-majority-slashable-liveness} is presented in Appendix~\ref{sec:appendix-security-proofs}.

\subsection{Proof Sketch for Corollary~\ref{cor:btc-sec-liveness}}

When the number of adversarial active validators is less than $n/2$, the \pos chains satisfy liveness by Theorem~\ref{thm:btc-honest-majority-slashable-liveness}.
In this case, after entering the \pos chains held by the clients, transactions soon appear in the checkpointed chains as the \pos chains are checkpointed regularly, implying the liveness of the protocol with slow finality.
Proof is given in Appendix~\ref{sec:appendix-security-proofs}.

\section{Proof of Theorem~\ref{thm:pos-non-slashable}}
\label{sec:appendix-slashable-safety}

\begin{proof}
Proof is by contradiction based on an indistinguishability argument between two worlds.
Towards contradiction, suppose there exists a \pos protocol $\PI$ that provides $\fL$-$\Tconfirm$-liveness and $\fA$-slashable-safety for some integers $\fL, \fA>0$ and $\Tconfirm<\infty$.

Let $n$ be the number of active validators at any given slot.
Let $P_0$, $P$, $Q'$ and $Q''$ denote disjoint sets of validators:
$P_0=\{\validator^0_i,i=1,..,n\}$, $P=\{\validator_i,i=1,..,n\}$, $Q'=\{\validator'_i,i=1,..,n\}$ and $Q''=\{\validator''_i,i=1,..,n\}$.
Let $T<\infty$ denote the time it takes for a validator to withdraw its stake after its withdrawal transaction is finalized by the \pos protocol.

The initial set of active validators is $P_0$.
The environment inputs transactions to the active validators so that the active validator set changes over time.
Suppose the active validator set becomes $P$ at some slot $t$ (\eg, $P=P_0$ at $t=0$).
We then consider the following four worlds:

\textbf{World 1:}
Validators in $P$ and $Q'$ are honest.
At slot $t$, $\Env$ inputs transactions $\tx'_i$, $i=1,..,n$, to the validators in $P$.
Here, $\tx'_i$ is the withdrawal transaction for $\validator_i$.
Validators in $P$ execute the \pos protocol, and record the consensus messages in their transcripts.
The environment $\mathcal{Z}$ replaces $\validator_i$ with $\validator'_i \in Q'$ as the new active validator.

At slot $t+\Tconfirm+T$, $(\Adv,\Env)$ spawns the client $\client_1$, which receives messages from the validators in $Q'$.
By $\Tconfirm$-liveness, for all $i \in [n]$, $\tx'_i \in \poschain^{\client_1}_{t+\Tconfirm+T}$.
Moreover, by slot $t+\Tconfirm+T$, all validators in $P$ have withdrawn their stake in $\client_1$'s view, and the set of active validators is $Q'$.

\textbf{World 2:}
Validators in $P$ and $Q''$ are honest.
At slot $t$, $\Env$ inputs transactions $\tx''_i$, $i=1,..,n$, to the validators in $P$.
Here, $\tx''_i$ is the withdrawal transaction for $\validator_i$.
Validators in $P$ execute the \pos protocol, and record the consensus messages in their transcripts.
The environment $\mathcal{Z}$ replaces each $\validator_i$ with $\validator''_i \in Q''$ as the new active validator.

At slot $t+\Tconfirm+T$, $(\Adv,\Env)$ spawns client $\client_2$, which receives messages from the validators in $Q''$.
By $\Tconfirm$-liveness, for all $i \in [n]$, $\tx''_i \in \poschain^{\client_2}_{t+\Tconfirm+T}$.
Moreover, by slot $t+\Tconfirm+T$, all validators in $P$ have withdrawn their stake in $\client_2$'s view, and the set of active validators is $Q''$.

\textbf{World 3:} 
Validators in $Q'$ are honest.
Validators in $P$ and $Q''$ are adversarial.
At slot $t$, $\Env$ inputs transactions $\tx'_i$, $i=1,..,n$, to the validators in $P$.
Validators in $P$ execute the \pos protocol, and record the consensus messages they observe in their transcripts.

Simultaneous with the execution above, $(\Adv, \Env)$ creates a simulated execution, where a different set of transactions, $\tx''_i$, $i \in [n]$, are input to the validators in $P$ at slot $t$.
In the simulated execution, $\mathcal{Z}$ replaces each $\validator_i$ with $\validator''_i \in Q''$ as the new active validator.
As in the real execution, validators in $P$ record the consensus messages of the simulated execution in their transcripts.

Finally, $(\Adv,\Env)$ spawns two clients $\client_1$ and $\client_2$ at slot $t+\Tconfirm+T$.
Here, $\client_1$ receives messages from the validators in $Q'$ whereas $\client_2$ receives messages from the validators in $Q''$.
Since the worlds 1 and 3 are indistinguishable by $\client_1$ except with negligible probability, for all $i \in [n]$, $\tx'_i \in \poschain^{\client_1}_{t+\Tconfirm+T}$ with overwhelming probability.
Since the worlds 2 and 3 are indistinguishable by $\client_2$ except with negligible probability, for all $i \in [n]$, $\tx''_i \in \poschain^{\client_2}_{t+\Tconfirm+T}$ with overwhelming probability.
Similarly, for all $i \in [n]$, $\tx'_i \notin \poschain^{\client_2}_{t+\Tconfirm+T}$, and $\tx''_i \notin \poschain^{\client_1}_{t+\Tconfirm+T}$.
Thus, $\poschain^{\client_1}_{t+\Tconfirm+T}$ and $\poschain^{\client_2}_{t+\Tconfirm+T}$ conflict with each other with overwhelming probability.
Moreover, at slot $t+\Tconfirm+T$, in the view of $\client_1$ and $\client_2$, the set of active validators are $Q'$ and $Q''$ respectively, and all validators in $P$ have withdrawn their stake.

As there is a safety violation and $\fA>0$, at least one validator must have become slashable in the view of both clients.
By definition of the forensic protocol, with overwhelming probability, a validator from the set $Q''$ becomes slashable in the clients' views as (i) the validators in $P$ have withdrawn their stake in the clients' views and (ii) those in $Q'$ are honest.

\textbf{World 4:} 
World 4 is the same as world 3, except that the validators in $Q'$ are adversarial, those in $Q''$ are honest, and the real and simulated executions are run with the transactions $\tx''_i$ and $\tx'_i$ respectively.
Given the messages received by the clients $\client_1$ and $\client_2$ from the validators, the worlds 3 and 4 are indistinguishable in their views except with negligible probability.
Thus, again a validator from $Q''$ becomes slashable in the clients' views in world 4 with non-negligible probability.
However, the validators in $Q''$ are honest in world 4, which is a contradiction with the definition of the forensic protocol.
\end{proof}

\section{Proof of Theorem~\ref{thm:data-limit}}
\label{sec:appendix-liveness-resilience}

\begin{proof}
\noindent
Proof is by contradiction based on an indistinguishability argument between two worlds.
Towards contradiction, suppose there exists a permissioned (or \pos) protocol $\PI$ with access to a data-limited timestamping service, that provides $\fS$-accountable-safety, and $\fL$-$\Tconfirm$-liveness for some integers $\fS>0$, $\fL \geq n/2$, and $\Tconfirm<\infty$.
Let $P$ and $Q$ denote two sets that partition the validators into two groups of size $\lceil n/2 \rceil$ and $\lfloor n/2 \rfloor$ respectively.
Consider the following worlds, where $\mathcal{Z}$ inputs externally valid bit strings, $\tx^P_i$, $i \in [\lceil n/2 \rceil]$, and $\tx^Q_j$, $j \in [\lfloor n/2 \rfloor]$, to the validators in $P$ and $Q$ respectively at the beginning of the execution.
Here, each validator $i$ in $P$ receives the unique string $\tx^P_i$, and each validator $j$ in $Q$ receives the unique string $\tx^Q_j$.
Each string consists of $m$ bits, where $m$ is a polynomial in the security parameter $\lambda$.

\textbf{World 1:}
There are two clients $\client_1$ and $\client_2$.
Validators in $P$ are honest, and those in $Q$ are adversarial.
In their heads, the adversarial validators simulate the execution of $\lfloor n/2 \rfloor$ honest validators that do not receive any messages from those in $P$ over the network.
They also do not send any messages to $P$ and $\client_1$, but reply to $\client_2$.

Validators in $Q$ send messages to the timestamping service $I$ as dictated by the protocol $\PI$.
There could be messages on $I$ sent by the validators in $P$ that require a response from those in $Q$. 
In this case, the validators in $Q$ reply as if they are honest validators and have not received any messages from those in $P$ over the network.

As $|Q| = \lfloor n/2 \rfloor \leq \fL$, by the $\fL$-liveness of $\PI$, clients $\client_1$ and $\client_2$ both output $\tx^P_i$, $i \in [\lceil n/2 \rceil]$ as part of their chains by slot $\Tconfirm$.
Since there can be at most $m\lfloor n/2 \rfloor-1$ bits of data on $I$, and $\tx^Q_j$, $j \in [\lfloor n/2 \rfloor]$, consist of $m\lfloor n/2 \rfloor$ bits in total, $\client_1$ does not learn and cannot output all of $\tx^Q_j$, $j \in [\lfloor n/2 \rfloor]$, as part of its chain by slot $\Tconfirm$, with overwhelming probability.

\textbf{World 2:}
There are again two clients $\client_1$ and $\client_2$.
Validators in $P$ are adversarial, and those in $Q$ are honest.
In their heads, the adversarial validators simulate the execution of the $\lceil n/2 \rceil$ honest validators from world 1, and pretend as if they do not receive any messages from those in $Q$ over the network.
They also do not send any messages to $Q$ and $\client_1$, but reply to $\client_2$.
They send the same messages to $I$ as those sent by the honest validators within world 1.

As $|P| = \lceil n/2 \rceil \leq \fL$, by the $\fL$-liveness of $\PI$, clients $\client_1$ and $\client_2$ both output $\tx^Q_j$, $j \in [\lfloor n/2 \rfloor]$, as part of their chains by slot $\Tconfirm$.
Since there can be at most $m\lfloor n/2 \rfloor-1$ bits of data on $I$, and $\tx^P_i$, $i \in [\lceil n/2 \rceil]$, consist of $m\lceil n/2 \rceil$ bits in total, $\client_1$ does not learn and cannot output all of $\tx^P_i$, $i \in [\lceil n/2 \rceil]$, as part of its chain by slot $\Tconfirm$, with overwhelming probability

The worlds 1 and 2 are indistinguishable by $\client_2$ in terms of the messages received from the validators and observed on the timestamping service, except with negligible probability.
Thus, it outputs the same chain containing $\tx^P_i$, $i \in [\lceil n/2 \rceil]$, and $\tx^Q_j$, $j \in [\lfloor n/2 \rfloor]$, in both worlds with overwhelming probability.
However, $\client_1$'s chain contains $\tx^P_i$, $i \in [\lceil n/2 \rceil]$, but not $\tx^Q_j$, $j \in [\lfloor n/2 \rfloor]$, in world 1, and $\tx^Q_j$, $j \in [\lfloor n/2 \rfloor]$, but not $\tx^P_i$, $i \in [\lceil n/2 \rceil]$, in world 2.
This implies that there is a safety violation in either world 1 or world 2 or both worlds with non-negligible probability.

Without loss of generality, suppose there is a safety violation in world 2.
In this case, an honest validator $\validator$ asks the other validators for their transcripts, upon which the adversarial validators in $P$ reply with transcripts that omit the messages received from the set $Q$. 
As $\fS > 0$, by invoking the forensic protocol with the transcripts received, $\validator$ identifies a non-empty subset $S \subseteq P$ of the adversarial validators, and outputs a proof that the validators in $S$ have violated the protocol $\PI$.
However, in this case, an adversarial validator in world 1 can emulate the behavior of $\validator$ in world 2, and ask the validators for their transcripts.
It can then invoke the forensic protocol with the transcripts, and output a proof that identifies the same subset $S \subseteq P$ of validators as protocol violators.
Since the two worlds are indistinguishable by $\client_2$ except with negligible probability, upon receiving this proof, it identifies the honest validators in $S \subseteq P$ as protocol violators in world 1 as well with non-negligible probability, which is a contradiction.
By the same reasoning, if the safety violation happened in world 1, an adversarial validator in world 2 can construct a proof accusing an honest validator in world 2 in $\client_2$'s view with non-negligible probability, again a contradiction.
\end{proof}

\section{Security Proofs}
\label{sec:appendix-security-proofs}

\begin{proof}[Proof of Proposition~\ref{prop:auxiliary-chain}]
Since \btc is safe with parameter $k$, without loss of generality, suppose $\powchain^{\client_1}_{r_1} \preceq \powchain^{\client_2}_{r_2}$ for two clients $\client_1$ and $\client_2$ and time slots $r_1$ and $r_2$.
Let $h_i$, $i \in [m_1]$, and $h_j$, $j \in [m_2]$, $m_1 \leq m_2$, denote the sequence of checkpoints in $\client_1$'s and $\client_2$'s views at slots $r_1$ and $r_2$ respectively.
Note that the sequence observed by $\client_1$ is a prefix of $\client_2$'s sequence.

Starting from the genesis block, let $B_1$ denote the first \pos block in $\aux^{\client_1}_{r_1}$ that is not available or finalized in $\client_2$'s view at slot $r_2$, and define $i_1$ as the index of the first valid checkpoint on $\powchain^{\client_1}_{r_1}$, whose block extends or is the same as $B_1$ on $\aux^{\client_1}_{r_1}$.
(If there is no such block $B_1$, $i_1 = \infty$.)
Similarly, let $B_2$ denote the first \pos block in $\aux^{\client_2}_{r_2}$ that is not available or finalized in $\client_1$'s view at slot $r_1$, and define $i_2$ as the index of the first valid checkpoint, whose block extends or is the same as $B_2$ on $\aux^{\client_2}_{r_2}$.
(If there is no such block $B_2$, $i_2 = \infty$.)

By the collision-resistance of the hash function, for any valid checkpoint $h_i$, $i \in [m_1]$, with index $i<\min(r_1,r_2)$, the condition at Line~\ref{line:btc1} of Alg.~\ref{alg.bitcoin.checkpointing} is true for $\client_1$ if and only if it is true for $\client_2$. 
Similarly, the clients must have updated their checkpointed chains with the same $L_i$ at Line~\ref{line:update} of Alg.~\ref{alg.bitcoin.checkpointing}.
This is because no checkpoint with index $i < \min(r_1,r_2)$, $i \in [m_1]$ triggers the stalling condition (Line~\ref{line:btc2}, Alg.~\ref{alg.bitcoin.checkpointing}).
Thus, if $i_1 = \infty$, $\aux^{\client_1}_{r_1} \preceq \aux^{\client_2}_{r_2}$.
However, if $i_1<\infty$, then $i_2 = \infty$, due to Line~\ref{line:btc2} of Alg.~\ref{alg.bitcoin.checkpointing} being triggered for $\client_2$.
Thus, if $i_1<i_2$, then $\aux^{\client_2}_{r_2} \prec \aux^{\client_1}_{r_1}$, and if $i_2 \leq i_1$, then $\aux^{\client_1}_{r_1} \preceq \aux^{\client_2}_{r_2}$.
Finally, by the safety of \btc with parameter $k$, $\powchain^{\client}_{r_1} \preceq \powchain^{\client}_{r_2}$ for any $r_2 \geq r_1$ and client $\client$.
Thus, $\aux^{\client}_{r_1} \preceq \aux^{\client}_{r_2}$.
\end{proof}

For the \pos protocol underlying Babylon to enforce accountability when there is a safety violation, it must ensure that the \pos blocks outputted by a client and the corresponding signatures are seen by at least one honest validator before the client goes offline; so that the honest validators can detect safety violations and create proofs accusing the adversarial validators.
Hence, in the proof below, we assume that the consensus messages and blocks in a clients' views are eventually observed by the validators before the client goes offline.

\begin{proof}[Proof of Theorem~\ref{thm:btc-only-slashable-safety}]
Suppose there are two clients $\client_1$, $\client_2$, and slots $r_1$, $r_2$ such that $\poschain^{\client_1}_{r_1}$ conflicts with $\poschain^{\client_2}_{r_2}$.
Starting from the genesis, let $B_1$ and $B_2$ denote the earliest conflicting blocks in $\poschain^{\client_1}_{r_1}$ and $\poschain^{\client_2}_{r_2}$ respectively.
As $B_1$ and $B_2$ share a common parent, they also share the same active validator set.

Let $r$ denote the first slot such that a valid checkpoint for a chain $L$ extending $B$ appears in the (confirmed) Bitcoin chain held by a client $\client$.
If there is no such slot, let $r = \infty$.
If $r < \infty$, define $b$ as the \btc block containing the checkpoint on $\powchain^{\client}_{r}$.
By synchrony and the safety of Bitcoin, $b$ appears in the Bitcoin chains of all awake clients (including honest validators) by slot $r+\Delta$ with the same prefix and contains the first valid checkpoint on those chains, whose \pos block extends $B$.
Since $B_1$ and $B_2$ conflict with each other, a block in $L$ also conflicts with at least one of the blocks $B_1$ and $B_2$ (By the collision-resistance of the hash function, the adversary cannot open the checkpoint within $b$ to conflicting chains).
Without loss of generality, suppose a block in $L$ conflicts with $B_1$.
Let $B'$ denote the first block in $L$ conflicting with $B_1$ and $r' \in [r,r+\Delta)$ denote the slot, in which $b$ appears in $\client_1$'s Bitcoin chain for the first time (if $\client_1$ is online). 
We next consider the following cases:

\paragraph{$r = \infty$}
The active validators for the blocks $B_1$ and $B_2$ cannot withdraw their stake in any client's view before a valid checkpoint for a chain extending $B$ appears in the (confirmed) Bitcoin chain held by the client, \ie, before slot $r$.
Moreover, the blocks $B_1$ and $B_2$ and the corresponding pre-commit signatures are eventually observed by an honest validator.
Since $B_1$ and $B_2$ are finalized and conflicting, and the \pos protocol satisfies $n/3$-accountable safety, the honest validator creates and sends a fraud proof (to Bitcoin) that identifies at least $n/3$ adversarial validators in the common active validator set of $B_1$ and $B_2$ as protocol violators.
This proof eventually appears in the Bitcoin chains held by all clients, implying that at least $n/3$ adversarial active validators become slashable in their views.
We hereafter assume that $r < \infty$.

\paragraph{$r_1 \geq r'+\Delta$}
If $L$ is available and finalized in $\client_1$'s view, then $\client_1$ would not have outputted $B_1$ as part of its \pos chain since $B_1$ conflicts with $B'$ (Line~\ref{line:btc1}, Alg.~\ref{alg.bitcoin.checkpointing}).
If $L$ contains a block that is unavailable or not finalized in $\client_1$'s view, then it would have triggered the stalling condition for $\client_1$, implying that $\client_1$ would not have outputted $B_1$ as part of its \pos chain (Line~\ref{line:btc1}, Alg.~\ref{alg.bitcoin.checkpointing}).
Thus, by contradiction, it must be that $r_1 < r'+\Delta$, and we hereafter assume $r_1 < r'+\Delta$.

\paragraph{The client $\client_1$ is online, and $L$ is available and finalized in $\client_1$'s view}
In this case, since $B_1$ and $B'$ are conflicting and finalized, and the \pos protocol satisfies $n/3$-accountable safety, $\client_1$ creates and sends a fraud proof (to Bitcoin) that identifies at least $n/3$ adversarial validators in the common active validator set of $B_1$ and $B'$ as protocol violators, by slot $r'+\Delta \leq r+2\Delta$.
By Proposition~\ref{prop:chain-growth}, this proof appears in the Bitcoin chains held by all clients $\client$ by slot $r''$ such that $|\powchain^{\client}_{r''}| = |\powchain^{\client}_{r}|+k$.  
Thus, at least $n/3$ adversarial active validators become slashable in all clients' views.

\paragraph{The client $\client_1$ is online, and $L$ is unavailable or not finalized in $\client_1$'s view}
In this case, the chain $L$ triggers the stalling condition for $\client_1$ (Line~\ref{line:btc2}, Alg.~\ref{alg.bitcoin.checkpointing}), in which case it sends a checkpoint to Bitcoin for its \pos chain at slot $r'+\Delta \leq r+2\Delta$.
By Proposition~\ref{prop:chain-growth}, this checkpoint appears in the Bitcoin chains held by all clients $\client$ by slot $r''$ such that $|\powchain^{\client}_{r''}| = |\powchain^{\client}_{r}|+k$.
Since there are now valid checkpoints for conflicting \pos blocks ($B_1$ and $B'$) with the same active validator set within $k$ Bitcoin blocks of each other, and the \pos protocol satisfies $n/3$-accountable safety, at least $n/3$ adversarial active validators become slashable in all clients' views (\cf condition 3 for withdrawals in Section~\ref{sec:btc-protocol}).

\paragraph{The client $\client_1$ has gone offline by slot $r+2\Delta$, and $L$ is available and finalized in an honest validator $\validator$'s view}
In this case, $\client_1$ must have sent its \pos chain including block $B_1$ and its pre-commit signatures to the validators by slot $r+2\Delta$, which are received by all honest validators by slot $r+3\Delta$.
Thus, both $B_1$ and $B'$ must have become available and finalized in $\validator$'s view by slot $r+3\Delta$.
Since these are conflicting blocks, and the \pos protocol satisfies $n/3$-accountable safety, $\validator$ creates and sends a fraud proof (to Bitcoin) that identifies at least $n/3$ adversarial validators in the common active validator set of $B_1$ and $B'$ as protocol violators, by slot $r+3\Delta$.
By Proposition~\ref{prop:chain-growth}, this proof appears in the Bitcoin chains held by all clients $\client$ by slot $r''$ such that $|\powchain^{\client}_{r''}| = |\powchain^{\client}_{r}|+k$, implying that at least $n/3$ adversarial active validators become slashable in all clients' views.

\paragraph{The client $\client_1$ has gone offline by slot $r+2\Delta$, and $L$ is unavailable or not finalized in all honest validators' views}
In this case, $\client_1$ must have sent its \pos chain including block $B_1$ and its pre-commit signatures to the validators by slot $r+2\Delta$, which are received by all honest validators by slot $r+3\Delta$.
If there is an honest validator $\validator$ whose \pos chain at slot $r+3\Delta$ has a block $B'$ conflicting with $B_1$, then, $\validator$ creates and sends a fraud proof (to Bitcoin) that identifies at least $n/3$ adversarial validators in the common active validator set of $B_1$ and $B'$ as protocol violators, by slot $r+3\Delta$.
By Proposition~\ref{prop:chain-growth}, this proof appears in the Bitcoin chains held by all clients $\client$ by slot $r''$ such that $|\powchain^{\client}_{r''}| = |\powchain^{\client}_{r}|+k$, implying that at least $n/3$ adversarial active validators become slashable in all clients' views.

On the other hand, suppose the \pos chains of all honest validators are consistent with that of $\client_1$.
Since $L$ contains a block that is unavailable or not finalized in the honest validators' views, it triggers the stalling condition (Line~\ref{line:btc2}, Alg.~\ref{alg.bitcoin.checkpointing}), in which case each honest validator sends a checkpoint to Bitcoin for its \pos chain at some slot in $[r+\Delta,r+3\Delta]$.
By Proposition~\ref{prop:chain-growth}, these checkpoints appear in the Bitcoin chains held by all clients $\client$ by slot $r''$ such that $|\powchain^{\client}_{r''}| = |\powchain^{\client}_{r}|+k$.
Moreover, these checkpoints are for \pos chains extending that of $\client_1$.
Since there are now valid checkpoints for conflicting \pos blocks ($B_1$ and $B'$) with the same active validator set within $k$ Bitcoin blocks of each other, and the \pos protocol satisfies $n/3$-accountable safety, at least $n/3$ adversarial active validators become slashable in all clients' views (\cf condition 3 for withdrawals in Section~\ref{sec:btc-protocol}).
\end{proof}

\begin{proof}[Proof of Theorem~\ref{thm:btc-only-liveness}]
By Theorem~\ref{thm:btc-only-slashable-safety}, Babylon 1.0 protocol satisfies $n/3$-slashable safety.
Hence, when the number of active adversarial validators is less than $n/3$ at all slots, it satisfies safety.
In this case, no valid checkpoint for unavailable or non-finalized blocks appears on Bitcoin.
Thus, clients never stop outputting new \pos blocks as part of their checkpointed and \pos chains.

Suppose a transaction $\tx$ is first input to an honest validator at some slot $r$.
Then, from slot $r$ on, each honest validator $\validator$ includes $\tx$ in its proposal block until it outputs a \pos block containing $\tx$.
Let $\client'$ be the client with the longest \pos chain among all clients at slot $r$.
By synchrony and the safety of the \pos protocol, for every client $\client$, $\poschain^{\client'}_r \preceq \poschain^{\client}_{r+\Delta}$.
Then, either $\poschain^{\client'}_r = \poschain^{\client}_{r+\Delta}$ for all clients $\client$, or there exists a client $\client$ such that $\poschain^{\client'}_r \prec \poschain^{\client}_{r+\Delta}$.

In the former case, every client agrees on the validator set at slot $r+\Delta$.
By \cite[Lemma 7]{tendermint}, there exists a finite $\Ttendermint = \mathrm{poly}(\lambda)$ such that if every client agrees on the validator set, a new block that extends $\poschain^{\client'}_r$ is finalized and becomes part of all clients' \pos chains by slot $r+\Ttendermint$, with overwhelming probability.
In the latter case, by synchrony, a new block that extends the longest \pos chain, thus all \pos chains held by the clients at slot $r$, is finalized in the view of $\client'$ and all other clients by slot $r+2\Delta$.

Finally, when the number of adversarial active validators is bounded by $n/3$, with probability at least $2/3$, each block finalized after slot $r$ must have been proposed by an honest validator.
Then, for any given integer $k>1$, by slot $r+(k+1)\max(\Ttendermint,2\Delta)$, the transaction $\tx$ will appear in all clients' \pos chains, except with probability $k \cdot \negl(\lambda) + (1/3)^k$.
Setting $k = \Theta(\lambda)$, it holds that $k \cdot \negl(\lambda) + (1/3)^k = \negl(\lambda)$.
Consequently, liveness is satisfied with parameter $\Tconfirm = \mathrm{poly}(\lambda)$, with overwhelming probability.
\end{proof}

\begin{proof}[Proof of Theorem~\ref{thm:btc-honest-majority-slashable-safety}]
Suppose there are two clients $\client_1$, $\client_2$, and slots $r_1$, $r_2$ such that $\poschain^{\client_1}_{r_1}$ conflicts with $\poschain^{\client_2}_{r_2}$.
Starting from the genesis, let $B_1$ and $B_2$ denote the earliest conflicting \pos blocks or bundles in $\poschain^{\client_1}_{r_1}$ and $\poschain^{\client_2}_{r_2}$ respectively.
Without loss of generality, let $r_1$ and $r_2$ be the first slots $B_1$ and $B_2$ appear in $\client_1$'s and $\client_2$'s \pos chains respectively.

By the safety of \btc, $\powchain^{\client_1}_{r_1}$ is a prefix of $\powchain^{\client_2}_{r_2}$ or vice versa with overwhelming probability.
By Proposition~\ref{prop:auxiliary-chain}, $\aux^{\client_1}_{r_1}$ is a prefix of $\aux^{\client_2}_{r_2}$ or vice versa.

We first consider the case, where at least one of the blocks is a bundle.
Without loss of generality, let $B_1$ be a bundle and $B$ denote the common parent of $B_1$ and $B_2$.
Let $b$ denote the \btc block with the liveness transaction that triggered the rollup mode, during which the valid bundle checkpoint $h_1$ for $B_1$ appeared in $\powchain^{\client_1}_{r_1}$.
At slot $r_1$, the prefix of $\client_1$'s \pos chain ending at $B_1$ consists of two parts: (i) a checkpointed chain outputted using the prefix of $\powchain^{\client_1}_{r_1}$ that ends at the $2k^{\text{th}}$ block extending $b$, (ii) bundles extending the checkpointed chain until $B_1$.

If $B$ is also a bundle, the next block in $\poschain^{\client_2}_{r_2}$ extending $B$, \ie $B_2$, has to be the same block as $B_1$ due to the consistency of $\powchain^{\client_1}_{r_1}$ and $\powchain^{\client_2}_{r_2}$.
However, as $B_2 \neq B_1$, $B$ cannot be a bundle.
If $B$ is not a bundle, it must be the last \pos block in $\client_1$'s checkpointed chain preceding $B_1$, implying that $B_1$ is the first bundle in $\poschain^{\client_1}_{r_1}$.
However, this again implies $B_1=B_2$, since $\client_1$ and $\client_2$ agree on the first block of the rollup mode whenever $\powchain^{\client_1}_{r_1}$ and $\powchain^{\client_2}_{r_2}$ are consistent.
As this is a contradiction, neither of the blocks $B_1$ or $B_2$ can be a bundle, with overwhelming probability.

Finally, if neither of $B_1$ and $B_2$ is a bundle, proof of slashable safety proceeds as given for Theorem~\ref{thm:btc-only-slashable-safety}.
\end{proof}

\begin{proof}[Proof of Theorem~\ref{thm:btc-honest-majority-slashable-liveness}]
As the number of honest active validators is over $n/2$ at all times, no valid \pos or bundle checkpoint for unavailable or non-finalized blocks or bundles appears on Bitcoin.
Hence, the clients do not stop outputting new \pos blocks or bundles as part of their \pos chains due to unavailable checkpoints (Lines~\ref{line:unavailable1} and~\ref{line:unavailable2} of Alg~\ref{alg.liveness.recovery} are never triggered).

Consider a transaction $\tx$ input to an honest validator at some slot $r$.
If $\tx$ does not appear in $\poschain^{\validator}_{r+\Ttendermint}$ in an honest validator $\validator$'s view, $\validator$ sends a liveness transaction to \btc containing $\tx$, at slot $r+\Ttendermint$.
Let $R = \mathrm{poly}(\lambda)$, denote the confirmation latency of \btc with parameter $k$.
Then, by the security of \btc, for all clients $\client$, the liveness transaction appears in $\powchain^{\client}_{r_1}$ within the same \btc block $b$, by slot $r_1 = r+\Ttendermint+R$, with overwhelming probability.

Once a client $\client$ observes $b$ become $k$ deep in its \btc chain, which happens by some slot less than $r_1+R$, it sends a checkpoint transaction for its \pos chain.
Subsequently, $b$ becomes at least $2k$ deep in $\client$'s \btc chain by some slot $r_2 \leq r_1+2R$.
In this case, there are two possibilities:
If $\tx \in \aux^{\client}_{r_2}$, then for all clients $\client'$, it holds that $\tx \in \poschain^{\client'}_{r_2+\Delta}$.
If $\tx \notin \aux^{\client}_{r_2}$, then all clients enter the rollup mode by slot $r_2+\Delta$.

Upon entering the rollup mode, an honest validator $\validator$ prepares a bundle containing $\tx$, which is viewed by all clients and signed by all honest validators by slot $r_2+2\Delta$.
Upon gathering these signatures, \ie, by slot $r_2+3\Delta$, $\validator$ sends a checkpoint for that bundle.
By the security of \btc, the checkpoint appears in the \btc chain of all clients at the same position by slot $r_3 = r_2 + 3\Delta + R$, with overwhelming probability.
Since the \pos protocol is accountable, an honest validator cannot be identified as a protocol violator and never becomes slashable in the view of any client.
Thus, the signatures within the bundle checkpoint from the honest validators are counted and suffice to pass the $n/2$ threshold, implying the validity of the bundle.
Consequently, $\tx \in \poschain^{\client}_{r_3}$ for all clients $\client$.

Setting $\Tconfirm = r_3 - r = 3\Delta + 4R + \Ttendermint = \mathrm{poly}(\lambda)$, we conclude the proof of $\Tconfirm$-liveness for the Babylon protocol.
\end{proof}

\begin{proof}[Proof of Corollary~\ref{cor:btc-sec-liveness}]
By Theorem~\ref{thm:btc-honest-majority-slashable-liveness}, if the number of adversarial active validators is less than $n/2$ at all times, the Babylon protocol of Section~\ref{sec:bitcoin-with-honesty} with fast finality satisfies $\Tconfirm$-liveness, where $\Tconfirm = \mathrm{poly}(\lambda)$.
Thus, if a transaction $\tx$ is input to an honest validator at some slot $r$, then for all clients $\client$ that follow the fast finality rule, $\tx$ will be in $\poschain^{\client}_{r+\Tconfirm}$.
Let $R = \mathrm{poly}(\lambda)$, denote the confirmation latency of \btc with parameter $k$.
Let $T = \mathrm{poly}(\lambda)$, denote an upper bound on the duration of epochs when there is no safety or liveness violation on Tendermint.

If $\tx \in \aux^{\client'}_{r+\Tconfirm}$ for a client $\client'$, then, by synchrony and the safety of the checkpointed chains, it appears within the checkpointed chains of all clients by slot $r+\Tconfirm+\Delta$.
On the other hand, $\tx$ might be in $\poschain^{\client}_{r+\Tconfirm}$ for all clients $\client$, but not in $\aux^{\client}_{r+\Tconfirm}$ for any client $\client$.
In this case, let $e$ denote the epoch of the \pos block containing $\tx$.

Suppose all clients eventually observe the same canonical \pos chain ending at the last block of epoch $e$, and no liveness transaction appears on the Bitcoin chain of any client by that time.
At the end of the epoch, an honest active validator $\validator$ sends a checkpoint transaction for its \pos chain, which happens at some slot $r'<r+\Tconfirm+T$.
Then, for any client $\client$, $\validator$'s checkpoint is in $\powchain^{\client}_{r'+R}$, with overwhelming probability.
Then, either $\tx \in \aux^{\client}_{r'+R}$ for all clients $\client$, or there is a checkpoint or fraud proof in the prefix of $\validator$'s checkpoint for conflicting \pos blocks.
In the latter case, the protocol enters the rollup mode as described in Section~\ref{sec:slashing}, in which case $\tx$ is eventually included within the checkpointed chains of all clients by slot $r'+R+\Tconfirm$.

Finally, if the above condition is violated, then there must have been a safety or liveness violation, or a liveness transaction must have appeared in the Bitcoin chain of a client, by slot $r+\Tconfirm+T$.
In this case, the process to enter the rollup mode is triggered by slot $r+\Tconfirm+T$, and $\tx$ is eventually included within the checkpointed chains of all clients by slot $r'+T+2\Tconfirm$.
Consequently, Babylon with slow finality satisfies $\Tconfirm'$-liveness, where $\Tconfirm' = T+R+2\Tconfirm = \mathrm{poly}(\lambda)$.
\end{proof}

\section{Accountability for Liveness Violations}
\label{sec:accountable-liveness}
In this section, we investigate if the concept of accountability can be extended to liveness violations.
Towards this goal, we first define the notion of \alr for blockchain protocols using the same formalism as accountable safety.

To ensure that the clients can detect liveness violations, we assume that they know and agree on the time when a transaction $\tx$ is first input to an honest validator.
In practice, disagreements on the input time of transactions can be bounded by $\Delta$ in a synchronous network by stipulating that the honest validators broadcast new transactions upon reception.
If a client observes that $\Tconfirm$-liveness is violated, \ie, a transaction input to an honest validator at some slot $r$ is not in the client's \pos chain by slot $t+\Tconfirm$, it invokes a forensic protocol by querying the validators.
The honest validators then send their transcripts to the client.
The forensic protocol takes these transcripts, and outputs a proof that identifies $f$ adversarial validators as protocol violators.
This proof is subsequently sent to all other clients, and serves as evidence that the identified validators have irrefutably violated the protocol rules.

\begin{definition}
\label{def:accountable-liveness}
A blockchain protocol is said to provide accountable liveness with resilience $f$ if when there is a liveness violation, (i) at least $f$ adversarial validators are identified by the forensic protocol as protocol violators, and (ii) no honest validator is identified, with overwhelming probability.
Such a protocol provides \emph{$f$-accountable-liveness}.
\end{definition}
$f$-accountable-liveness implies $f$-liveness, and as such is a stronger property.

Unfortunately, accountable liveness is not possible even with the help of Bitcoin, or for that matter any timestamping service, unless the entirety of the blocks are uploaded to the service. 
In particular, the adversary can execute unaccountable liveness attacks whenever it controls more than half the validators.
\begin{theorem}
\label{thm:data-limit-liveness}
Consider a PoS or permissioned protocol with $n$ validators in a $\Delta$ synchronous network such that the protocol provides $\fS$-accountable-safety for some $\fS>0$, and has access to a timestamping service.
Suppose each validator is given an externally valid input of $m$ bits by the environment $\mathcal{Z}$, but the number of bits written to the timestamping service is less than $m\lfloor n/2 \rfloor-1$.
Then, the protocol cannot provide $\fA$-$\Tconfirm$-accountable-liveness for any $\fA>0$ and $\Tconfirm<\infty$ if the number of active adversarial validators can be $n/2$ or more.
\end{theorem}

Although a PoS protocol $\PI$ can provide accountable safety and liveness given access to a timestamping service with no data limitation, such a service is equivalent to a secure external blockchain protocol, which can be directly used to order transactions instead of $\PI$.
This would in turn satisfy safety and liveness for any number of adversarial validators on the protocol $\PI$ as long as there is one honest validator to forward the transactions to the timestamping service.

\begin{proof}[Proof of Theorem~\ref{thm:data-limit-liveness}]
\noindent
Proof is by contradiction based on an indistinguishability argument between two worlds.
Towards contradiction, suppose there exists a permissioned (or \pos) protocol $\PI$, with access to a data-limited timestamping service $I$, that provides $\fS$-accountable-safety, and $\fA$-$\Tconfirm$-accountable-liveness for some integers $\fA,\fS>0$ and $\Tconfirm<\infty$.
Let $\fL$ denote the $\Tconfirm$-liveness resilience of $\PI$.

We first analyze the case $\fL \geq n/2$.
Let $P$ and $Q$ denote two sets that partition the validators into two groups of size $\lceil n/2 \rceil$ and $\lfloor n/2 \rfloor$ respectively.
Consider the following worlds, where $\mathcal{Z}$ inputs externally valid bit strings, $\tx^P_i$, $i \in [\lceil n/2 \rceil]$, and $\tx^Q_j$, $j \in [\lfloor n/2 \rfloor]$, to the validators in $P$ and $Q$ respectively at the beginning of the execution.
Here, each validator $i$ in $P$ receives the unique string $\tx^P_i$, and each validator $j$ in $Q$ receives the unique string $\tx^Q_j$.
Each string consists of $m$ bits, where $m$ is a polynomial in the security parameter $\lambda$.

\textbf{World 1:}
There are two clients $\client_1$ and $\client_2$.
Validators in $P$ are honest, and those in $Q$ are adversarial.
In their heads, the adversarial validators simulate the execution of $\lfloor n/2 \rfloor$ honest validators that do not receive any messages from those in $P$ over the network.
They also do not send any messages to $P$ and $\client_1$, but reply to $\client_2$.

Validators in $Q$ send messages to the timestamping service $I$ as dictated by the protocol $\PI$.
There could be messages on $I$ sent by the validators in $P$ that require a response from those in $Q$. 
In this case, the validators in $Q$ reply as if they are honest validators and have not received any messages from those in $P$ over the network.

As $|Q| = \lfloor n/2 \rfloor \leq \fL$, by the $\fL$-liveness of $\PI$, clients $\client_1$ and $\client_2$ both output $\tx^P_i$ for all $i \in [\lceil n/2 \rceil]$ as part of their chains by slot $\Tconfirm$.
Since there can be at most $m\lfloor n/2 \rfloor-1$ bits of data on $I$, and $\tx^Q_j$, $j \in [\lfloor n/2 \rfloor]$, consists of $m\lfloor n/2 \rfloor$ bits in total, $\client_1$ does not learn and cannot output all of $\tx^Q_j$, $j \in [\lfloor n/2 \rfloor]$, as part of its chain by slot $\Tconfirm$ with overwhelming probability.

\textbf{World 2:}
There are again two clients $\client_1$ and $\client_2$.
Validators in $P$ are adversarial, and those in $Q$ are honest.
In their heads, the adversarial validators simulate the execution of the $\lceil n/2 \rceil$ honest validators from world 1 and pretend as if they do not receive any messages from those in $Q$ over the network.
They also do not send any messages to $Q$ and $\client_1$, but reply to $\client_2$.
They send the same messages to $I$ as those sent by the honest validators within world 1.

As $|P| = \lceil n/2 \rceil \leq \fL$, by the $\fL$-liveness of $\PI$, clients $\client_1$ and $\client_2$ both output $\tx^Q_j$, $j \in [\lfloor n/2 \rfloor]$, as part of their chains by slot $\Tconfirm$.
Since there can be at most $m\lfloor n/2 \rfloor-1$ bits of data on $I$, and $\tx^P_i$, $i \in [\lceil n/2 \rceil]$, consists of $m\lceil n/2 \rceil$ bits in total, $\client_1$ does not learn and cannot output all of $\tx^P_i$, $i \in [\lceil n/2 \rceil]$, as part of its chain by slot $\Tconfirm$ with overwhelming probability

The worlds 1 and 2 are indistinguishable by $\client_2$ in terms of the messages received from the validators and observed on the timestamping service, except with negligible probability.
Thus, it outputs the same chain containing $\tx^P_i$, $i \in [\lceil n/2 \rceil]$, and $\tx^Q_j$, $j \in [\lfloor n/2 \rfloor]$, in both worlds with overwhelming probability.
However, $\client_1$'s chain contains $\tx^P_i$, $i \in [\lceil n/2 \rceil]$, but not all of $\tx^Q_j$, $j \in [\lfloor n/2 \rfloor]$, in world 1, and $\tx^Q_j$, $j \in [\lfloor n/2 \rfloor]$, but not all of $\tx^P_i$, $i \in [\lceil n/2 \rceil]$, in world 2.
This implies that there is a safety violation in either world 1 or world 2 or both worlds with non-negligible probability.

Without loss of generality, suppose there is a safety violation in world 2.
In this case, $\client_1$ asks the validators for their transcripts, upon which the adversarial validators in $P$ reply with transcripts that omit the messages received from the set $Q$ over the network.
As $\fS > 0$, by invoking the forensic protocol with the transcripts received, $\client_1$ identifies a non-empty subset $S \subseteq P$ of the adversarial validators, and outputs a proof that the validators in $S$ have violated the protocol $\PI$.
However, in this case, an adversarial validator in world 1 can emulate the behavior of $\client_1$ in world 2, and ask the validators for their transcripts.
It can then invoke the forensic protocol with the transcripts, and output a proof that identifies the same subset $S \subseteq P$ of validators as protocol violators.
Since the two worlds are indistinguishable by $\client_2$ except with negligible probability, upon receiving this proof, it identifies the honest validators in $S \subseteq P$ as protocol violators in world 1 as well, with non-negligible probability, which is a contradiction.
By the same reasoning, if the safety violation happened in world 1, an adversarial validator in world 2 can construct a proof accusing an honest validator in world 2 in $\client_2$'s view with non-negligible probability, again a contradiction.

We next analyze the case $\fL < n/2$.

\textbf{World 3:}
There are two clients, $\client_1$ and $\client_2$.
Validators in $P$ are honest, and those in $Q$ are adversarial.
Adversarial validators behave as described in world 1.
Since there can be at most $m\lfloor n/2 \rfloor-1$ bits of data on $I$, and $\tx^Q_j$, $j \in [\lfloor n/2 \rfloor]$, consists of $m\lfloor n/2 \rfloor$ bits in total, $\client_1$ does not learn and cannot output all of $\tx^Q_j$, $j \in [\lfloor n/2 \rfloor]$, as part of its chain by slot $\Tconfirm$.
As there are at least $\fL$ adversarial validators, either of the following cases can happen:
\begin{itemize}
    \item $\client_1$ outputs $\tx^P_i$, $i \in [\lceil n/2 \rceil]$, as part of its chain by slot $\Tconfirm$.
    \item $\client_1$ does not output $\tx^P_i$, $i \in [\lceil n/2 \rceil]$, by slot $\Tconfirm$.
\end{itemize}

\textbf{World 4:}
There are again two clients, $\client_1$ and $\client_2$.
Validators in $P$ are adversarial, and those in $Q$ are honest.
Adversarial validators behave as described in world 2.
Since there can be at most $m\lfloor n/2 \rfloor-1$ bits of data on $I$, and $\tx^P_i$, $i \in [\lceil n/2 \rceil]$, consists of $m\lceil n/2 \rceil$ bits in total, $\client_1$ does not learn and cannot output all of $\tx^P_i$, $i \in [\lceil n/2 \rceil]$, as part of its chain by slot $\Tconfirm$.
As there are at least $\fL$ adversarial validators, either of the following cases can happen:
\begin{itemize}
    \item $\client_1$ outputs $\tx^Q_j$, $j \in [\lfloor n/2 \rfloor]$, as part of its chain by slot $\Tconfirm$.
    \item $\client_1$ does not output $\tx^Q_j$, $j \in [\lfloor n/2 \rfloor]$, by slot $\Tconfirm$.
\end{itemize}

The worlds 3 and 4 are indistinguishable by $\client_2$ in terms of the messages received from the validators and observed on the timestamping service, except with negligible probability.
Thus, it outputs the same (potentially empty) chain in both worlds by slot $\Tconfirm$ with overwhelming probability.
Suppose $\client_2$ does not output all of $\tx^P_i$, $i \in [\lceil n/2 \rceil]$, as part of its chain by $\Tconfirm$.
As this implies a violation of $\Tconfirm$-liveness in world 3, it asks the validators for their transcripts, upon which the adversarial validators in $Q$ reply with transcripts that omit the messages received from the set $P$ over the network.
As $\fA > 0$, by invoking the forensic protocol with the received transcripts, $\client_1$ identifies a non-empty subset $S \subseteq Q$ of the adversarial validators, and outputs a proof that the validators in $S$ have violated the protocol $\PI$.
However, in this case, an adversarial validator in world 4 can emulate the behavior of $\client_2$ in world 3, and ask the validators for their transcripts.
It can then invoke the forensic protocol with the transcripts, and output a proof that identifies the same subset $S \subseteq Q$ of validators as protocol violators.
Since the two worlds are indistinguishable by $\client_2$ except with negligible probability, upon receiving this proof, it would identify the honest validators in $S \subseteq Q$ as protocol violators in world 4 as well, with non-negligible probability, which is a contradiction.
By the same reasoning, if $\client_2$ does not output all of $\tx^Q_j$, $j \in [\lfloor n/2 \rfloor]$, as part of its chain by slot $\Tconfirm$, the adversary can construct a proof accusing an honest validator in world 3 in $\client_2$'s view with non-negligible probability, again a contradiction.

Next, suppose $\client_1$ does not output all of $\tx^P_i$, $i \in [\lceil n/2 \rceil]$, as part of its chain by slot $\Tconfirm$ in world 3.
As this implies a violation of $\Tconfirm$-liveness in world 3, it asks the validators for their transcripts, upon which the adversarial validators in $Q$ reply with transcripts that omit the messages received from the set $P$ over the network. 
As $\fA > 0$, by invoking the forensic protocol with the transcripts received, $\client_1$ identifies a non-empty subset $S \subseteq Q$ of the adversarial validators, and outputs a proof that the validators in $S$ have violated the protocol $\PI$.
However, in this case, an adversarial validator in world 4 can emulate the behavior of $\client_1$ in world 3, and ask the validators for their transcripts.
It can then invoke the forensic protocol with the transcripts, and output a proof that identifies the same subset $S \subseteq Q$ of validators as protocol violators.
Since the two worlds are indistinguishable by $\client_2$ except with negligible probability, upon receiving this proof, it would identify the honest validators in $S \subseteq Q$ as protocol violators in world 4 with non-negligible probability, which is a contradiction.
By the same reasoning, if $\client_1$ does not output $\tx^Q_j$, $j \in [\lfloor n/2 \rfloor]$, as part of its chain by slot $\Tconfirm$, the adversary can construct a proof accusing an honest validator in world 3 in $\client_2$'s view with non-negligible probability, again a contradiction.

Finally, if $\client_1$ outputs $\tx^P_i$, $i \in [\lceil n/2 \rceil]$, and $\tx^Q_j$, $j \in [\lfloor n/2 \rfloor]$, as part of its chain in worlds 3 and 4 respectively, and $\client_2$ outputs both $\tx^P_i$, $i \in [\lceil n/2 \rceil]$, and $\tx^Q_j$, $j \in [\lfloor n/2 \rfloor]$, as part of its chain in both worlds, we reach a contradiction with the statement $f_s>0$ by the same reasoning presented for the worlds 1 and 2.
Consequently, under the given conditions, no permissioned or PoS protocol can provide a positive accountable safety and liveness resilience simultaneously even with access to a timestamping service.
\end{proof}

\section{Security under Partial Synchrony}
\label{sec:appendix-partial-synchrony}

In this section, we extend our results to a new model that considers network partitions among the \pos validators. 
Note that \btc would not be secure in a partially synchronous network, where the adversary can delay the protocol messages arbitrarily before a global stabilization time (GST)\footnote{Here, GST is determined by the adversary and unknown to the clients and honest validators. All messages by honest validators are delivered within $\Delta$ rounds after GST.}.
Thus, we cannot directly prove slashable safety or liveness of our protocols under partial synchrony.
Moreover, replacing \btc with a different protocol that is secure under partial synchrony does not suffice to ensure slashable safety, since slashable safety requires the timely confirmation of the fraud proofs sent by the \pos validators.
Therefore, we analyze Babylon 1.0 in a model, where the communication among the \pos validators and clients happen over a partially-synchronous network, yet \btc is assumed to be secure and Proposition~\ref{prop:chain-growth} holds, \ie, messages exchanged between validators and \btc miners reach their destination within $\Delta$ time.
Although the validators can achieve a synchronous communication (albeit with a large delay upper bound) by posting their messages to \btc, this would not only overload \btc, which is data-limited, but is unnecessary altogether, as we argue that Babylon 1.0 achieves slashable safety and liveness in this model without any change.

Since clients that go offline might not be able to communicate the \pos blocks in their views to the other clients before GST, for the proof below, we assume that the clients are always online.
Note that satisfying slashable safety remains challenging in this setting due to partial synchrony.

Proof of Theorem~\ref{thm:btc-only-slashable-safety} under a partially-synchronous network for \pos validators and clients remains the same as the proof in Appendix~\ref{sec:appendix-security-proofs}, except that it does not consider offline clients.
This is because the proof does not use any synchrony assumption for the messages exchanged among \pos validators and clients.
On the other hand, Theorem~\ref{thm:btc-only-liveness} is updated to reflect the fact that liveness can only be guaranteed after GST.

\begin{theorem}[Liveness under Partial Synchrony]
\label{thm:btc-only-liveness-psync}
Suppose \btc is secure with parameter $k$ with overwhelming probability, and the number of adversarial active validators is less than $n/3$ at all times.
Then, the Babylon 1.0 protocol with fast finality (Section~\ref{sec:btc-protocol}) satisfies $\Tconfirm$-liveness after GST with overwhelming probability, where $\Tconfirm=\Theta(\lambda)$. 
\end{theorem}

\begin{proof}[Proof of Theorem~\ref{thm:btc-only-liveness-psync}]
By Theorem~\ref{thm:btc-only-slashable-safety}, Babylon 1.0 protocol satisfies $n/3$-slashable safety.
Hence, when the number of active adversarial validators is less than $n/3$ at all slots, it satisfies safety.
In this case, no valid checkpoint for unavailable or non-finalized blocks appears on \btc after GST.
Thus, clients never stop outputting new \pos blocks as part of their checkpointed and \pos chains after GST.
Then, liveness follows from the liveness of the \pos protocol (\ie, Tendermint~\cite{tendermint}) after GST.

\end{proof}

\end{document}